\documentclass{report}

\usepackage{amsmath}
\usepackage{amsthm}
\usepackage{amssymb}
\usepackage{stmaryrd}
\usepackage{mathpartir}
\usepackage{hyperref}
\usepackage[backend=bibtex]{biblatex}

\usepackage{tikz}
\usetikzlibrary{matrix}

\title{Polymorphic Bottom-Up Weighted Relational Programming}
\author{Dmitri Volkov}

\newtheorem{theorem}{Theorem}
\newtheorem{lemma}{Lemma}

\theoremstyle{definition}
\newtheorem{definition}{Definition}

\newcommand{\den}[1]{\llbracket #1 \rrbracket}

\newcommand{\wt}{\mathrel{\mathrm{wt}}}
\newcommand{\vt}{\mathrel{\mathrm{vt}}}

\DeclareMathOperator{\shell}{shell}
\DeclareMathOperator{\holes}{holes}
\DeclareMathOperator{\envshell}{envshell}
\DeclareMathOperator{\envholes}{envholes}

\newcommand{\eqpat}{\leftrightharpoons}
\newcommand{\eqpatD}{\eqpat_{\Delta}}

\DeclareMathOperator{\enforceeqpatD}{enforce_{\eqpatD}}
\DeclareMathOperator{\compile}{compile}

\begin{document}

\maketitle

\thispagestyle{empty}
\begin{center}
    Accepted by Luddy Graduate Faculty, Indiana University, in partial
    fulfillment of the requirements for the degree of Master of
    Science.
\end{center}
\vspace{5\baselineskip}
\setlength{\unitlength}{1pc}
\rightline{Master's Committee\hfill\line(1,0){14}}
\rightline{Chung-chieh Shan, Ph.D.}
\vspace{5\baselineskip}
\rightline{\line(1,0){14}}
\rightline{Lawrence S. Moss, Ph.D.}
\vspace{5\baselineskip}
\leftline{April 15, 2026}

\begin{abstract}
	This work presents a new approach for implementing polymorphism for bottom-up relational languages, without monomorphization.
	We begin by introducing semiringKanren, a bottom-up weighted relational programming language.
	We extend this base language to support polymorphism.
	We describe a new method to compile polymorphic semiringKanren programs into non-polymorphic ones, based on \emph{equality patterns} and \emph{large-enough instances} of polymorphic relations.
	We prove the correctness of this method.
	Finally, we consider existing work and suggest directions for future research.
\end{abstract}

\tableofcontents

\chapter{Introduction}

Relational programming (also known as logic programming) is arguably the most expressive and powerful paradigm for constructing computations.
Where most languages express programs as sequences of instructions or as mathematical functions, relational languages express programs with the full power of mathematical relations.
Applications of relational programming include natural language processing\cite{mooney1997} and program generation\cite{byrd2012}, both of which can be difficult to express in other paradigms.

Relational programming languages are usually implemented either via \emph{top-down} search, or \emph{bottom-up} fact collection.
Top-down relational programming languages include Prolog and miniKanren; bottom-up relational languages include datalog and SQL.
The top-down approach can effectively handle recursive data, but can also be prone to getting stuck searching an ineffectual branch.
The bottom-up approach may not be able to handle recursive data, and may use space less efficiently, but has a clearer semantics (especially for weighted relations) and can avoid many of the traps of top-down languages.

We use the term \emph{polymorphism} to describe code that can operate on different types of data.
When programmers want to processes different types of data in similar ways, polymorphism allows them to do so without manually re-writing similar code for each particular type.
Polymorphism is generally classified as either \emph{parametric} or \emph{ad-hoc}.
In parametric polymorphism, unknown types are expressed as \emph{type variables}, and direct operations on values of unknown types are prohibited.
Ad-hoc polymorphism does allow operations on values of unknown types, and these operations behave differently depending on the underlying actual types of the values.
In this work, we focus on parametric polymorphism.

The notion of a bottom-up polymorphic relational programming language can feel paradoxical.
How can we collect facts like ``\(x\) is related to \(y\)'' when we do not know the domains \(x\) and \(y\) are drawn from?
One approach might be to track which types for \(x\) and \(y\) are used in practice, and create specialized versions of polymorphic relations for these types.
This approach is known as \emph{monomorphization}.
Unfortunately, this can substantially increase the size of programs, and does not generalize.
What if we want to call a relation on a new type that does not yet have specialized relations?

In this work, we present semiringKanren, a bottom-up weighted relational programming language supporting polymorphism \emph{without} monomorphization.
We begin by introducing the base semiringKanren language (which does not support polymorphism) to build familiarity and establish the foundational concepts.
We then show how the syntax and type system of this base language can be extended to support polymorphism.
We briefly describe a monomorphizing semantics for this language, before introducing a non-monomorphizing compilation technique from the polymorphic language to the base language.
In particular, we show how \emph{equality patterns} can be used to extract the underlying structure of \emph{large-enough instances} of polymorphic relations
to handle polymorphic relation calls at different types.
We conclude the section with correctness proofs of this approach.
We compare semiringKanren with existing work, focusing particularly on other polymorphic bottom-up weighted relational languages.
Finally, we consider areas for future exploration and development.

We believe this new approach to polymorphism can offer improvements in the broader space of relational programming.
In particular, we hope it may offer efficiency gains, and enable ``separate evaluation'' of polymorphic programs.

\chapter{semiringKanren}

We introduce the semiringKanren language and the underlying principles.
The semiringKanren language consists of relations defined in terms of goals, which may take in values inhabiting algebraic data types.
Relations take arguments, and semiringKanren computes weights for every possible concrete value assignment to arguments.

The semiringKanren language as presented here differs from prior work~\cite{volkov2025}.
Rather than having primitive ``constructor'' relations, this version has an extended value language allowed within unification and disunification.
We believe this makes programs easier to write, and simplifies the implementation of algorithms occurring later in this work.

\section{Example Programs}

We introduce semiringKanren by example.

\subsection{Relations}

In semiringKanren, programs are made up of relations, which in turn are made up of goals.
Relations take zero or more typed argument variables, and either \emph{succeed} or \emph{fail} for each possible assignment of concrete values to those variables.
We use \(\texttt{defrel}\) to define new relations.
\begin{equation}
\begin{aligned}
	& \texttt{(defrel (coin-flip (coin : (Sum Unit Unit)))} \\
	& \quad \texttt{(disj} \\
	& \quad\quad \texttt{(== coin (left sole))} \\
	& \quad\quad \texttt{(== coin (right sole))))} \\
\end{aligned}
\end{equation}
Here, we define a relation \(\texttt{coin-flip}\) that takes a single argument of type \(\texttt{(Sum Unit Unit)}\).
The type \(\texttt{(Sum Unit Unit)}\) has two values: \(\texttt{(left sole)}\) and \(\texttt{(right sole)}\).
We use \(\texttt{(left sole)}\) to represent tails, and \(\texttt{(right sole)}\) to represent heads.
This definition showcases two types of goals: the connective \(\texttt{disj}\) and the primitive relation \(\texttt{==}\).
The \(\texttt{disj}\) goal succeeds when either of its subgoals succeeds.
The primitive relation \(\texttt{==}\) succeeds when its two arguments are equal.
Thus, this relation succeeds either when it is called with the argument \(\texttt{(left sole)}\), or with the argument \(\texttt{(right sole)}\).

Note that many relational programming languages have a distinction between relations and \emph{queries} on those relations.
This distinction is not necessary in semiringKanren because relation evaluation already directly computes the satisfying argument values.

\subsection{Weighted Relations}

Beyond success and failure, semiringKanren relations can have different \emph{weights}.
We can add weight to a program branch with the \(\texttt{factor}\) goal.
We can modify the previous example to express a weighted unfair coin flip:
\begin{equation}
\begin{aligned}
	& \texttt{(defrel (unfair-coin-flip (coin : (Sum Unit Unit)))} \\
	& \quad \texttt{(disj} \\
	& \quad \quad \texttt{(conj (factor 0.7) (== coin (left sole)))} \\
	& \quad \quad \texttt{(conj (factor 0.3) (== coin (right sole)))))} \\
\end{aligned}
\end{equation}
We can think of \(\texttt{conj}\) as succeeding when both of its subgoals succeed.
Nonzero weights represent success, and a weight of \(0\) represents failure.
When working with weights, \(\texttt{conj}\) generalizes to multiplication and \(\texttt{disj}\) generalizes to addition.
Hence, this relation assigns the weight \(0.7\) to \(\texttt{(left sole)}\) (representing heads), and the weight \(0.3\) to \(\texttt{(right sole)}\) (representing tails).

\subsection{Recursion}

Within a relation, we can call other relations or even recursively call the current relation.
For example, we can simulate a fair coin flip with unfair coins~\cite{vonneumann1963}.
For this process, we flip an unfair coin twice.
If the two coin flips have the same result, we repeat the process.
Otherwise, we go with the result of the first coin flip.
\begin{equation}
\begin{aligned}
	& \texttt{(defrel (fair-coin-clip (coin : (Sum Unit Unit)))} \\
	& \quad \texttt{(fresh ((coin1 : (Sum Unit Unit))} \\
	& \quad\quad\quad\quad\quad\; \texttt{(coin2 : (Sum Unit Unit)))} \\
	& \quad\quad \texttt{(conj} \\
	& \quad\quad\quad \texttt{(unfair-coin-flip coin1)} \\
	& \quad\quad\quad \texttt{(unfair-coin-flip coin2)} \\
	& \quad\quad\quad \texttt{(disj} \\
	& \quad\quad\quad\quad \texttt{(conj} \\
	& \quad\quad\quad\quad\quad \texttt{(== coin1 coin2)} \\
	& \quad\quad\quad\quad\quad \texttt{(fair-coin-flip coin)))} \\
	& \quad\quad\quad\quad \texttt{(conj} \\
	& \quad\quad\quad\quad\quad \texttt{(=/= coin1 coin2)} \\
	& \quad\quad\quad\quad\quad \texttt{(== coin coin1)))))} \\
\end{aligned}
\end{equation}
Here, we use the \(\texttt{fresh}\) goal, which introduces \emph{fresh variables}, and sums over the weights of its subgoal for each concrete value the variables can take.
This can be thought of as a disjunction over all concrete values inhabited by the fresh variables.
We also use the disequality primitive relation \(\texttt{=/=}\), which tests that its arguments are not equal.

The semiringKanren languages handles recursion by finding the \emph{least fixed point} of relations.
This means it repeatedly evaluates each relation, using the weights from the previous evaluation whenever there is a call.
Throughout this process, relations gradually accumulate more information.
We refer to this as \emph{bottom-up evaluation}.

\subsection{Semirings}

So far we have been working with the real numbers \(\mathbb{R}\), where \(\texttt{disj}\) denotes real number addition, and \(\texttt{conj}\) denotes real number multiplication.
However, semiringKanren can evaluate programs over any \emph{commutative semiring}.
A commutative semiring is a set equipped with \(+\) and \(\times\) operations,
both of which are commutative (\(a + b = b + a\)), associative (\(a + (b + c) = (a + b) + c\)), and obey distributivity \(a \times (b + c) = (a \times b) + (a \times c)\).
We also require this set to have a multiplicative identity~\(1\), and an additive identity/multiplicative annihilator~\(0\).
Sometimes the term \emph{rig}, as in ``ri\underline{n}g without \underline{n}egation,'' is used.
Besides the real numbers equipped with addition and multiplication \((\mathbb{R}, +, *,0,1)\),
other semirings include the \emph{boolean semiring}, equipped with logical ``or'' and logical ``and'' \((\mathbb{B},\vee,\wedge,\bot,\top)\),
and the \emph{min-tropical semiring}: the real numbers and infinity equipped with minimum and addition \((\mathbb{R}^{\infty},\min,+,\infty,0)\).

To show the utility of different semirings, we present the example of transitive closure of a relation.
Transitive closure can be thought of as ``reachability by chaining.''
Given a relation \(a \mathrel{R} b\), for each element \(x\), we want to find all \(y\) such that \(x \mathrel{R} \dots \mathrel{R} y\).
We can represent the relation we want to process as a graph:
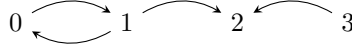
\begin{center}
    \begin{tikzpicture}[>=stealth]
        \node (m) [matrix of nodes, column sep=3em] { 0 & 1 & 2 & 3 \\ };
        \draw [->, bend left] (m-1-1) to (m-1-2);
        \draw [->, bend left] (m-1-2) to (m-1-1);
        \draw [->, bend left] (m-1-2) to (m-1-3);
        \draw [<-, bend left] (m-1-3) to (m-1-4);
    \end{tikzpicture}
\end{center}
Here, our relation operates on numbers \(0\) through \(3\).
For convenience, we call this type \(4\), and represent the members of this type with \(\mathtt{0_4},\mathtt{1_4},\mathtt{2_4},\mathtt{3_4}\).
In practice, we create this type using sum of four \(\texttt{Unit}\)s.
Assuming we give each edge in the graph some weight \(w\), we can write this relation in semiringKanren as follows:
\begin{equation}
\begin{aligned}
	& \texttt{(defrel (graph ((x : 4) (y : 4)))} \\
	& \quad \texttt{(disj} \\
	& \quad\quad \texttt{(conj (== x \(\mathtt{0_4}\)) (== y \(\mathtt{1_4}\)) (factor \(w\)))} \\
	& \quad\quad \texttt{(conj (== x \(\mathtt{1_4}\)) (== y \(\mathtt{0_4}\)) (factor \(w\)))} \\
	& \quad\quad \texttt{(conj (== x \(\mathtt{1_4}\)) (== y \(\mathtt{2_4}\)) (factor \(w\)))} \\
	& \quad\quad \texttt{(conj (== x \(\mathtt{3_4}\)) (== y \(\mathtt{2_4}\)) (factor \(w\)))))} \\
\end{aligned}
\end{equation}
We can now express reachability/transitive closure as a recursive relation: 
\begin{equation}
\begin{aligned}
	& \texttt{(defrel (connect ((x : 4) (y : 4)))} \\
	& \quad \texttt{(disj} \\
	& \quad\quad \texttt{(graph x y)} \\
	& \quad\quad \texttt{(fresh ((z : 4))} \\
	& \quad\quad\quad \texttt{(conj} \\
	& \quad\quad\quad\quad \texttt{(connect x z)} \\
	& \quad\quad\quad\quad \texttt{(connect z y)))))} \\
\end{aligned}
\end{equation}
When we evaluate this over the boolean semiring with \(w = \top\), we get reachability:
\begin{center}
	\begin{tabular}{l|cccc}
		& \(y = \mathtt{0_4}\) & \(y = \mathtt{1_4}\) & \(y = \mathtt{2_4}\) & \(y = \mathtt{3_4}\) \\
		\hline
		\(x = \mathtt{0_4}\) &\(\top\) & \(\top\) & \(\top\) & \(\bot\) \\
		\(x = \mathtt{1_4}\) &\(\top\) & \(\top\) & \(\top\) & \(\bot\) \\
		\(x = \mathtt{2_4}\) &\(\bot\) & \(\bot\) & \(\bot\) & \(\bot\) \\
		\(x = \mathtt{3_4}\) &\(\bot\) & \(\bot\) & \(\top\) & \(\bot\)
	\end{tabular}
\end{center}
When we evaluate this over the min-tropical semiring with \(w = 1\), we get the shortest path lengths:
\begin{center}
	\begin{tabular}{l|cccc}
		& \(y = \mathtt{0_4}\) & \(y = \mathtt{1_4}\) & \(y = \mathtt{2_4}\) & \(y = \mathtt{3_4}\) \\
		\hline
		\(x = \mathtt{0_4}\) &\(2\) & \(1\) & \(2\) & \(\infty\) \\
		\(x = \mathtt{1_4}\) &\(1\) & \(2\) & \(1\) & \(\infty\) \\
		\(x = \mathtt{2_4}\) &\(\infty\) & \(\infty\) & \(\infty\) & \(\infty\) \\
		\(x = \mathtt{3_4}\) &\(\infty\) & \(\infty\) & \(1\) & \(\infty\)
	\end{tabular}
\end{center}

\section{Syntax and Typing}

\begin{figure}
\[
\begin{array}{lrcl}
	\textrm{Programs} & P & ::= & \mathit{Rel}\dots \\
	\textrm{Relations} & Rel & ::= & \texttt{(defrel (\(R\) (\(x:\tau\)) \(\dots\)) \(g\)) } \\
	\textrm{Goals} & g & ::= & \texttt{(conj \(g\) \(g\))} \\
	& & | & \texttt{(disj \(g\) \(g\))} \\
	& & | & \texttt{(fresh ((\(x:\tau\))) \(g\))} \\
	& & | & \texttt{(== \(v\) \(v\))} \\
	& & | & \texttt{(=/= \(v\) \(v\))} \\
	& & | & \texttt{(\(R\) \(v\) \(\dots\))} \\
	& & | & \texttt{(factor \(k\))} \\
	\textrm{Values} & v & ::= & \texttt{sole} \\
	& & | & \texttt{(left\(_{\texttt{(Sum \(\tau_1\) \(\tau_2\))}}\) \(v\))} \\
	& & | & \texttt{(right\(_{\texttt{(Sum \(\tau_1\) \(\tau_2\))}}\) \(v\))} \\
	& & | & \texttt{(pair \(v\) \(v\))} \\
	& & | & x \\
	\textrm{Types} & \tau & ::= & \texttt{Unit} \\
	& & | & \texttt{(Sum \(\tau_1\) \(\tau_2\))} \\
	& & | & \texttt{(Pair \(\tau_1\) \(\tau_2\))} \\
	\textrm{Relation names} & R & & \\
	\textrm{Variables} & x,y,z & & \\
	\textrm{Weights} & r & \in & \mathbb{K}
\end{array}
\]
	\caption{Syntax for semiringKanren.}
	\label{fig:semiringkanren-syntax}
\end{figure}

We give the syntax of semiringKanren, omitting some syntactic sugar, in figure \ref{fig:semiringkanren-syntax}.
Here, \(\texttt{conj}\) and \(\texttt{disj}\) must have two subgoals, and \(\texttt{fresh}\) can only introduce a single variable at a time.

\begin{figure}
\begin{mathpar}
	\infer{ }{
		\Delta \vdash \texttt{sole} : \texttt{Unit}
	}
	\and
	\infer{
		\Delta \vdash v : \tau_1
	}
	{
		\Delta \vdash \texttt{(left\(_{\texttt{(Sum \(\tau_1\) \(\tau_2\))}}\) \(v\))} : \texttt{(Sum \(\tau_1\) \(\tau_2\))}
	}
	\and
	\infer{
		\Delta \vdash v : \tau_2
	}
	{
		\Delta \vdash \texttt{(right\(_{\texttt{(Sum \(\tau_1\) \(\tau_2\))}}\) \(v\))} : \texttt{(Sum \(\tau_1\) \(\tau_2\))}
	}
	\and
	\infer{
		\Delta \vdash v_1 : \tau_1
		\and
		\Delta \vdash v_2 : \tau_2
	}
	{
		\Delta \vdash \texttt{(pair \(v_1\) \(v_2\))} : \texttt{(Prod \(\tau_1\) \(\tau_2\))}
	}
	\and
	\infer{
	}
	{
		\Delta, x : \tau \vdash x : \tau
	}
\end{mathpar}
\caption{Inference rules for value typing judgement.}
\label{fig:value-types}
\end{figure}

\begin{figure}
\begin{mathpar}
	\infer{
		\Gamma;\Delta \vdash g_1 \wt
		\and
		\Gamma;\Delta \vdash g_2 \wt
	}{
		\Gamma;\Delta \vdash \texttt{(conj \(g_1\) \(g_2\))} \wt
	}
	\and
	\infer{
		\Gamma;\Delta \vdash g_1 \wt
		\and
		\Gamma;\Delta \vdash g_2 \wt
	}{
		\Gamma;\Delta \vdash \texttt{(disj \(g_1\) \(g_2\))} \wt
	}
	\and
	\infer{
		\tau \vt
		\and
		\Gamma;\Delta,x:\tau \vdash g \wt
	}{
		\Gamma;\Delta \vdash \texttt{(fresh ((\(x:\tau\))) \(g\))} \wt
	}
	\and
	\infer{
		\Delta \vdash v_1 : \tau
		\and
		\Delta \vdash v_2 : \tau
	}{
		\Gamma; \Delta \vdash \texttt{(== \(v_1\) \(v_2\))} \wt
	}
	\and
	\infer{
		\Delta \vdash v_1 : \tau
		\and
		\Delta \vdash v_2 : \tau
	}{
		\Gamma; \Delta \vdash \texttt{(=/= \(v_1\) \(v_2\))} \wt
	}
	\and
	\infer{
		\Delta \vdash \vec v : \vec \tau
	}{
		\Gamma, (R : \vec\tau \rightarrow); \Delta \vdash \texttt{(\(R\) \(\vec v\))} \wt
	}
	\and
	\infer{ }{
		\Gamma; \Delta \vdash \texttt{(factor \(r\))} \wt
	}
\end{mathpar}
\caption{Inference rules for goal typing judgement.}
\label{fig:goal-types}
\end{figure}

\begin{figure}
\begin{mathpar}
	\infer{
		\Gamma ;\vec x:\vec \tau \vdash g \wt
	}{
		\Gamma \vdash \texttt{(defrel (\(R\) (\(\vec x : \vec \tau\))) \(g\)))} \wt
	}
	\and
	\infer{
		(R_1 : \vec \tau_1 \rightarrow), \dots, (R_N : \vec \tau_N \rightarrow) \vdash
		\texttt{(defrel (\(R_1\) (\(\vec x_1 : \vec \tau_1\))) \(g_1\))} \wt
		\\\\
		\vdots
		\\\\
		(R_1 : \vec \tau_1 \rightarrow), \dots, (R_N : \vec \tau_N \rightarrow) \vdash
		\texttt{(defrel (\(R_N\) (\(\vec x_N : \vec \tau_N\))) \(g_1\))} \wt
	}{
		\texttt{(defrel (\(R_1\) (\(\vec x_1:\vec \tau_1\))) \(g_1\)))}, \dots, \\\\
		\texttt{(defrel (\(R_N\) (\(\vec x_N:\vec \tau_N\))) \(g_N\)))} \wt
	}
\end{mathpar}
\caption{Inference rules for relation and program typing judgements.}
\label{fig:relation-types}
\end{figure}

Figure \ref{fig:value-types} shows the type system for semiringKanren values,
figure \ref{fig:goal-types} shows the type system for goals, and figure \ref{fig:relation-types} shows the type system for relations and programs.
For values, the typing judgement \(\Delta \vdash v : \tau\) means that the value \(v\) has type \(\tau\) under the value context \(\Delta\), which tracks the types of variables.
Note we have annotated the \(\texttt{left}\) and \(\texttt{right}\) sum type constructors with full sum types, to disambiguate when type checking.
In a practical implementation, this is not necessary.
Each entry in \(\Delta\) has the form \(x : \tau\), where \(x\) is a variable (from either a relation argument, or \(\texttt{fresh}\)), and \(\tau\) is a type.
For goals, we use the judgement \(\Gamma; \Delta \vdash g \wt\) to indicate that the goal \(g\) is well-typed under value context \(\Delta\) and relation context \(\Gamma\).
The \(\Gamma\) context tracks the types of relations.
Each entry in \(\Gamma\) has the form \((R : \tau_1, \dots, \tau_n \to)\) to indicate that the relation with name \(R\) takes \(n\) arguments with types \(\tau_1, \dots, \tau_n\).
For relations, we use the similar judgement \(\Gamma \vdash Rel \wt\), omitting the \(\Delta\) context because it is not needed (there are no variables to bind).
For programs, where there is no larger context, we simply use \(P \wt\).

Note that our relation typing judgement is ``the relation is well-typed under a context'' rather than ``the relation has type \(\dots\).''
Relations are required to specify all needed typing information, so it only is necessary to check that their bodies are compatible.
Because the types of all relations are known ahead-of-time, we avoid potential issues of recursive dependencies when type-checking relation calls.
We use the notation \((\vec x : \vec\tau)\) as shorthand for \((x_1 : \tau_1) \; \dots \; (x_n : \tau_n)\), or similar.

\section{Semantics}

\begin{figure}
\begin{align*}
	\den{\texttt{Unit}} &= \{\texttt{sole}\} \\
	\den{\texttt{(Sum \(\tau_1\) \(\tau_2\))}} &= \{\texttt{(left \(v\))} \mid v \in \den{\tau_1}\} \cup \{\texttt{(right \(v\))} \mid v \in \den{\tau_2}\} \\
	\den{\texttt{(Prod \(\tau_1\) \(\tau_2\))}} &= \{\texttt{(pair \(v_1\) \(v_2\))} \mid v_1 \in \den{\tau_1}, v_2 \in \den{\tau_2}\} \\
\end{align*}
\caption{Denotational semantics for types.}
\label{fig:types-semantics}
\end{figure}

Each type \(\tau\) denotes a finite set of values, as shown in figure \ref{fig:types-semantics}.
We can use values as indices into an array.
One way this can be done is by ordering the values lexicographically, then using their position in the list as an index drawn from \(\mathbb{N}\).

\begin{definition}
	We define the \emph{size of type \(\tau\)}, notated \(|\tau|\), to be the number of values of type \(\tau\).
	\begin{align*}
		|\tau| = | \den{\tau} |
	\end{align*}
\end{definition}

\begin{figure}
\begin{align*}
	\den{\texttt{sole}} (\delta) &= \texttt{sole} \\
	\den{\texttt{(left \(v\))}} (\delta) &= \texttt{(left \(\den{v} (\delta) \))} \\
	\den{\texttt{(right \(v\))}} (\delta) &= \texttt{(right \(\den{v} (\delta) \))} \\
	\den{\texttt{(pair \(v_1\) \(v_2\))}} (\delta) &= \texttt{(pair \(\den{v_1} (\delta) \) \(\den{v_2} (\delta)\))} \\
	\den{x} (\delta) &= \delta(x)
\end{align*}
\caption{Denotational semantics for values.}
\label{fig:values-semantics}
\end{figure}

The denotation for values, as shown in figure \ref{fig:values-semantics}, is simply the values themselves but with substitutions applied to variables.
We use the term \emph{concrete values} for values that do not include variables.
The denotation for values turns arbitrary values (which may contain variables) into concrete values.
To this end, we take a \emph{value environment} \(\delta\), consisting of entries of the form \(x \mapsto v\), where \(x\) is a variable name and \(v\) is a concrete value.
For \(\delta\) to be valid, each value entry \(x \mapsto v\) should have a matching type entry \(x : \tau\) in \(\Delta\), where \(\cdot \vdash v : \tau\).
In the denotational semantics for values, when we reach a variable, we apply \(\delta\) to get its concrete value.

\begin{figure}
\begin{align*}
	\den{\texttt{(conj \(g_1\) \(g_2\))}} (\gamma;\delta) &= \den{g_1} (\gamma;\delta) \times \den{g_2} (\gamma;\delta) \\
	\den{\texttt{(disj \(g_1\) \(g_2\))}} (\gamma;\delta) &= \den{g_1} (\gamma;\delta) + \den{g_2} (\gamma;\delta) \\
	\den{\texttt{(fresh ((\(x : \tau\))) \(g\))}} (\gamma;\delta) &= \sum_{v \in \den{\tau}} \den{g} (\gamma;\delta,x \mapsto v) \\
	\den{\texttt{(== \(v_1\) \(v_2\))}} (\gamma;\delta) &=
	\begin{cases}
		1 & \textrm{if } \den{v_1} (\delta) = \den{v_2} (\delta) \\
		0 & \textrm{if } \den{v_1} (\delta) \ne \den{v_2} (\delta)
	\end{cases} \\
	\den{\texttt{(=/= \(v_1\) \(v_2\))}} (\gamma;\delta) &=
	\begin{cases}
		1 & \textrm{if } \den{v_1} (\delta) \ne \den{v_2} (\delta) \\
		0 & \textrm{if } \den{v_1} (\delta) = \den{v_2} (\delta)
	\end{cases} \\
	\den{\texttt{(\(R\) \(\vec v\))}} (\gamma;\delta) &= \gamma(R)(\den{\vec v}(\delta))
	\textrm{ where } \vec v : \vec\tau \\
	\den{\texttt{(factor \(r\))}} (\gamma;\delta) &= r
\end{align*}
\caption{Denotational semantics for goals.}
\label{fig:goals-semantics}
\end{figure}

The most interesting semiringKanren semantics is that of goals.
For the denotation of goals, we take a \emph{relation environment} \(\gamma\) mapping relation names to precomputed relations (whose types are given in \(\Gamma\)), and a value environment \(\delta\).
The relational environment \(\gamma\) is valid if for any \((R : \tau_1,\dots,\tau_n \to) \in \Gamma\),
the relation in \(\gamma\) with relation name \(R\) maps indices drawn from \(\den{\tau_1} \times \dots \times \den{\tau_n}\) to weights drawn from \(\mathbb{K}\).

The denotation of goals as shown in figure \ref{fig:goals-semantics} computes a weight depending on the relation environment and current value environment.
Goals can be understood as operating on arrays indexed by \(\delta\) and parameterized by \(\gamma\).
For example, \(\texttt{conj}\) and \(\texttt{disj}\) denote element-wise multiplication and addition, and \(\texttt{(factor \(r\))}\) denotes the constant array with the weight \(r\) in each entry.
When \(x\) and \(y\) are different variables, \(\texttt{(== \(x\) \(y\))}\) produces an array with ones along the diagonal of the \(x\) and \(y\) dimensions,
and \(\texttt{(=/= \(x\) \(y\))}\) produces an array with ones on the off-diagonal and zeros along the diagonal.
The \(\texttt{fresh}\) goal sums along one dimension of a rank-\((n + 1)\) array (corresponding to the dimension of the new variable) to produce a rank-\(n\) array.

\begin{figure}
\begin{gather*}
	\den{\texttt{(defrel (\(R\) (\(\vec x : \vec \tau\)) g)}} (\gamma) =
	(\vec i) \mapsto \den{g} (\gamma; \vec x \mapsto \den{\vec i}) \\
	\den{\textit{Rel}_1, \dots, \textit{Rel}_Q} = \;
	\mathrel{\textrm{fix}} \gamma \mathrel{\textrm{in}} R_1 \mapsto \den{\textit{Rel}_1} (\gamma), \dots, R_Q \mapsto \den{\textit{Rel}_Q} (\gamma)
\end{gather*}
\caption{Denotational semantics for relations and programs.}
\label{fig:relation-semantics}
\end{figure}

We define the denotation of a relation to be a map from elements drawn from the types of its arguments, to weights.
Thus, the elements of \(\gamma\) are precomputed relation denotations.
We compute these denotations by evaluating the goal structure of each relation for all possible values inhabiting its argument types.
We now define the denotation of a complete program to be the least fixed point of the denotation of all relations.
We use \(\gamma\) to repeatedly compute the denotation of each relation to get a new \(\gamma\), and repeat this process until \(\gamma\) stabilizes.
In practice, \(\gamma\) starts out consisting of relations that always fail.
Then, the fixpointing process can be thought of as iteratively deriving new information.
These semantics are shown in figure \ref{fig:relation-semantics}.

\chapter{Polymorphic semiringKanren}

\section{Polymorphic Relations}

We extend semiringKanren with \emph{polymorphic relations} that operate on values of unknown types (represented with \emph{type variables}).
We introduce polymorphic relations by example.

\subsection{Equal}

Consider the following relation:
\begin{equation}
\begin{aligned}
	& \texttt{(defrel (equal (\(x : \alpha\)) (\(y : \alpha\)))} \\
	& \quad \texttt{(== \(x\) \(y\)))}
\end{aligned}
\end{equation}
This relation takes in two values of the same variable type, and checks that they are equal.
We can call this relation at different types.
For example, following call succeeds:
\begin{equation}
	\texttt{(equal sole sole)} 
\end{equation}
while the following call fails:
\begin{equation}
\begin{aligned}
	& \texttt{(equal}\\
	& \quad \texttt{(pair (left sole) (right sole))} \\
	& \quad \texttt{(pair (right sole) (left sole)))} \\
\end{aligned}
\end{equation}
In the first example, \(\alpha\) gets mapped to the type \(\texttt{Unit}\).
In the second example, \(\alpha\) gets mapped to the type \(\texttt{(Prod (Sum Unit Unit) (Sum Unit Unit))}\).
Note that the primitive relations \(\texttt{==}\) and \(\texttt{=/=}\) are polymorphic, even in the base language:
they can operate on any type, including variable types, as long as both arguments have the same type.

\subsection{Sum-Swap}

Continuing to a larger example:
\begin{equation}
\begin{aligned}
	& \texttt{(defrel (sum-swap (\(x : \texttt{(Sum \(\alpha\) \(\beta\))}\)) (\(y : \texttt{(Sum \(\beta\) \(\alpha\))}\)))} \\
	& \quad \texttt{(disj} \\
	& \quad \quad \texttt{(fresh ((\(a : \alpha\)))} \\
	& \quad \quad \quad \texttt{(conj} \\
	& \quad \quad \quad \quad \texttt{(== \(x\) (left \(a\)))} \\
	& \quad \quad \quad \quad \texttt{(== \(y\) (right \(a\)))))} \\
	& \quad \quad \texttt{(fresh ((\(b : \beta\)))} \\
	& \quad \quad \quad \texttt{(conj} \\
	& \quad \quad \quad \quad \texttt{(== \(x\) (right \(b\)))} \\
	& \quad \quad \quad \quad \texttt{(== \(y\) (left \(b\)))))))} \\
\end{aligned}
\end{equation}
This relation replaces top-level \(\texttt{left}\)s in a sum type with \(\texttt{right}\)s, and vice versa.
This relation uses multiple type variables, and includes type variables within larger type structures.
Assuming we call \(\texttt{sum-swap}\) with \(\alpha = \texttt{Unit}\) and \(\beta = \texttt{(Sum Unit Unit)}\),
\begin{equation}
	\texttt{(sum-swap (left sole) \(q\))} 
\end{equation}
succeeds when \(q = \texttt{(right sole)}\), and
\begin{equation}
	\texttt{(sum-swap (right (right sole)) \(q\))} 
\end{equation}
succeeds when \(q = \texttt{(left (right sole))}\).
When we have multiple type variables, they can be mapped to different types.

\subsection{Option-Map}

As a ``real-world'' motivating example, consider the following:
\begin{equation}
\begin{aligned}
	& \texttt{(defrel (option-map} \\
	& \quad \texttt{(\(f : \texttt{(Prod \(\alpha\) \(\beta\))}\)) (\(x : \texttt{(Sum Unit \(\alpha\))}\)) (\(y : \texttt{(Sum Unit \(\beta\))}\)))} \\
	& \quad \texttt{(disj} \\
	& \quad \quad \texttt{(conj} \\
	& \quad \quad \quad \texttt{(== \(x\) (left sole))} \\
	& \quad \quad \quad \texttt{(== \(y\) (left sole)))} \\
	& \quad \quad \texttt{(fresh ((\(a : \alpha\)) (\(b : \beta\)))} \\
	& \quad \quad \quad \texttt{(conj} \\
	& \quad \quad \quad \quad \texttt{(== \(x\) (right \(a\)))} \\
	& \quad \quad \quad \quad \texttt{(== \(f\) (pair \(a\) \(b\)))} \\
	& \quad \quad \quad \quad \texttt{(== \(y\) (right \(b\)))))))} \\
\end{aligned}
\end{equation}
The type \(\texttt{(Sum Unit \(\alpha\))}\) can be thought of as an ``option type'' where \(\texttt{(right \(v\))}\) can be thought of carrying a value,
and \(\texttt{(left sole)}\) can be though of not carrying a value, or ``null.''
Functional analogs of \(\texttt{option-map}\) are common in typed functional programming languages when working with optional values.
The relation \(\texttt{option-map}\) takes a relation call encoded as a pair, and applies it to the ``value'' part of \(x\) and \(y\) if possible.
For example, \(\texttt{option-map}\) may be called as follows:
\begin{equation}
\begin{aligned}
	& \texttt{(fresh ((\(h : \texttt{(Sum Unit Unit)}\)) (\(k : \texttt{Sum Unit Unit}\)))} \\
	& \quad \texttt{(conj} \\
	& \quad \quad \texttt{(sum-swap \(h\) \(k\))} \\
	& \quad \quad \texttt{(option-map (pair \(h\) \(k\)) (right (left sole)) \(q\))))}
\end{aligned}
\end{equation}
This succeeds for \(q = \texttt{(right (right sole))}\).
Here, both \(\alpha\) and \(\beta\) are mapped to \(\texttt{(Sum Unit Unit)}\) within the \(\texttt{option-map}\) call.

\subsection{Two-Valued}

We present one more example:
\begin{equation}
\begin{aligned}
	& \texttt{(defrel (two-valued (\(x : \alpha\)))} \\
	& \quad \texttt{(fresh ((\(y : \alpha\)))} \\
	& \quad \quad \texttt{(=/= \(x\) \(y\))))}
\end{aligned}
\end{equation}
This relation succeeds when it can find a value of the argument's type disequal from the argument.
For example, assuming \(\alpha\) gets mapped to \(\texttt{Sum Unit Unit}\) the following succeeds:
\begin{equation}
	\texttt{(two-valued (left sole))}
\end{equation}
because it can find \(y\) with the value \(\texttt{(right sole)}\).
However, consider calling \(\texttt{two-valued}\) with \(\alpha\) mapped to \(\texttt{Unit}\):
\begin{equation}
	\texttt{(two-valued sole)}
\end{equation}
This call fails, because the only value of type \(\texttt{Unit}\) is \(\texttt{sole}\), and \(\texttt{(=/= sole sole)}\) fails.
Success/failure depends on the type of the argument of \(\texttt{two-valued}\) rather than the value.

\section{Roadmap}

The obvious way to implement polymorphism is to generate multiple instances of each relation for each type variable mapping they are called with.
This is known as \emph{monomorphization}, and is historically the only way polymorphism has been implemented for bottom-up relational languages.
With semiringKanren, we present a new method of implementing polymorphism based on \emph{large-enough instances} of polymorphic relations,
where polymorphic relation calls with arbitrary type mappings are compiled to only use a single type mapping, with extra code generated at call sites to reconstruct the desired types.

We first present the syntax and typing rules for semiringKanren extended with polymorphic relations.
Next, we provide a semantics for polymorphic semiringKanren based on monomorphization.
Finally, we show how polymorphic relation calls can be compiled to use large-enough instances, and prove the correctness of this approach.

\section{Syntax and Typing}

We expand the syntax and typing of minimal semiringKanren to support polymorphic relations.
We give the syntax of polymorphic semiringKanren in figure \ref{fig:poly-syntax}.

\begin{figure}
\[
\begin{array}{lrcl}
	\textrm{Programs} & P & ::= & \mathit{Rel}\dots \\
	\textrm{Relations} & Rel & ::= & \texttt{(defrel (\(R\) \(\forall \alpha \dots\) . (\(x:\tau\)) \(\dots\)) \(g\)) } \\
	\textrm{Goals} & g & ::= & \texttt{(conj \(g\) \(g\))} \\
	& & | & \texttt{(disj \(g\) \(g\))} \\
	& & | & \texttt{(fresh ((\(x:\tau\)) \(\dots\)) \(g\))} \\
	& & | & \texttt{(== \(v\) \(v\))} \\
	& & | & \texttt{(=/= \(v\) \(v\))} \\
	& & | & \texttt{(\(R\) \(v\) \(\dots\))} \\
	& & | & \texttt{(factor \(k\))} \\
	\textrm{Values} & v & ::= & \texttt{sole} \\
	& & | & \texttt{(left\(_{\texttt{(Sum \(\tau_1\) \(\tau_2\))}}\) \(v\))} \\
	& & | & \texttt{(right\(_{\texttt{(Sum \(\tau_1\) \(\tau_2\))}}\) \(v\))} \\
	& & | & \texttt{(pair \(v\) \(v\))} \\
	& & | & x \\
	\textrm{Types} & \tau & ::= & \texttt{Unit} \\
	& & | & \texttt{(Sum \(\tau_1\) \(\tau_2\))} \\
	& & | & \texttt{(Pair \(\tau_1\) \(\tau_2\))} \\
	& & | & \alpha \\
	\textrm{Relation names} & R & & \\
	\textrm{Variables} & x,y,z & & \\
	\textrm{Type variables} & \alpha,\beta & & \\
	\textrm{Weights} & r & \in & \mathbb{K}
\end{array}
\]
\caption{Syntax for polymorphic semiringKanren}
\label{fig:poly-syntax}
\end{figure}

\begin{figure}
\begin{mathpar}
	\infer{ }{\Delta \vdash \texttt{Unit} \vt}
	\and
	\infer{
		\Delta \vdash \tau_1 \vt \and \Delta \vdash \tau_2 \vt
	}{
		\Delta \vdash \texttt{(Sum \(\tau_1\) \(\tau_2\))} \vt
	}
	\and
	\infer{
		\Delta \vdash \tau_1 \vt \and \Delta \vdash \tau_2 \vt
	}{
		\Delta \vdash \texttt{(Prod \(\tau_1\) \(\tau_2\))} \vt
	}
	\and
	\infer{ }{\Delta,\alpha : * \vdash \alpha \vt}
\end{mathpar}
\caption{Inference rules for polymorphic type validity judgement.}
	\label{fig:poly-types-valid}
\end{figure}

\begin{figure}
\begin{mathpar}
	\infer{ }{
		\Delta \vdash \texttt{sole} : \texttt{Unit}
	}
	\and
	\infer{
		\Delta \vdash v : \tau_1
	}
	{
		\Delta \vdash \texttt{(left\(_{\texttt{(Sum \(\tau_1\) \(\tau_2\))}}\) \(v\))} : \texttt{(Sum \(\tau_1\) \(\tau_2\))}
	}
	\and
	\infer{
		\Delta \vdash v : \tau_2
	}
	{
		\Delta \vdash \texttt{(right\(_{\texttt{(Sum \(\tau_1\) \(\tau_2\))}}\) \(v\))} : \texttt{(Sum \(\tau_1\) \(\tau_2\))}
	}
	\and
	\infer{
		\Delta \vdash v_1 : \tau_1
		\and
		\Delta \vdash v_2 : \tau_2
	}
	{
		\Delta \vdash \texttt{(pair \(v_1\) \(v_2\))} : \texttt{(Prod \(\tau_1\) \(\tau_2\))}
	}
	\and
	\infer{
		\Delta, x : \tau \vdash \tau \vt
	}
	{
		\Delta, x : \tau \vdash x : \tau
	}
\end{mathpar}
\caption{Inference rules for polymorphic value typing judgement.}
	\label{fig:poly-value-types}
\end{figure}

\begin{figure}
\begin{mathpar}
	\infer{
		\Gamma;\Delta \vdash g_1 \wt
		\and
		\Gamma;\Delta \vdash g_2 \wt
	}{
		\Gamma;\Delta \vdash \texttt{(conj \(g_1\) \(g_2\))} \wt
	}
	\and
	\infer{
		\Gamma;\Delta \vdash g_1 \wt
		\and
		\Gamma;\Delta \vdash g_2 \wt
	}{
		\Gamma;\Delta \vdash \texttt{(disj \(g_1\) \(g_2\))} \wt
	}
	\and
	\infer{
		\Delta \vdash \tau \vt
		\and
		\Gamma;\Delta,x:\tau \vdash g \wt
	}{
		\Gamma;\Delta \vdash \texttt{(fresh ((\(x:\tau\))) \(g\))} \wt
	}
	\and
	\infer{
		\Delta \vdash v_1 : \tau
		\and
		\Delta \vdash v_2 : \tau
	}{
		\Gamma; \Delta \vdash \texttt{(== \(v_1\) \(v_2\))} \wt
	}
	\and
	\infer{
		\Delta \vdash v_1 : \tau
		\and
		\Delta \vdash v_2 : \tau
	}{
		\Gamma; \Delta \vdash \texttt{(=/= \(v_1\) \(v_2\))} \wt
	}
	\and
	\infer{
		\Delta \vdash \vec v : \sigma(\vec\tau)
	}{
		\Gamma, (R : \forall \vec\alpha . \vec\tau \rightarrow); \Delta \vdash \texttt{(\(R\) \(\vec v\))} \wt
	}
	\and
	\infer{ }{
		\Gamma; \Delta \vdash \texttt{(factor \(r\))} \wt
	}
\end{mathpar}
\caption{Inference rules for polymorphic goal typing judgement.}
	\label{fig:poly-goal-types}
\end{figure}

\begin{figure}
\begin{mathpar}
	\infer{
		\Gamma ; \vec \alpha:*, \vec x:\vec \tau \vdash g \wt
	}{
		\Gamma \vdash \texttt{(defrel (\(R\) \(\forall \vec \alpha\).(\(\vec x : \vec \tau\))) \(g\)))} \wt
	}
	\and
	\infer{
		(R_1 : \forall \vec \alpha_1 . \vec \tau_1 \rightarrow), \dots, (R_N : \forall \vec \alpha_N . \vec \tau_N \rightarrow) \vdash
		\texttt{(defrel (\(R_1\) \(\forall \vec \alpha_1\).(\(\vec x_1 : \vec \tau_1\))) \(g_1\))} \wt
		\\\\
		\vdots
		\\\\
		(R_1 : \forall \vec \alpha_1 . \vec \tau_1 \rightarrow), \dots, (R_N : \forall \vec \alpha_N . \vec \tau_N \rightarrow) \vdash
		\texttt{(defrel (\(R_N\) \(\forall \vec \alpha_N\).(\(\vec x_N : \vec \tau_N\))) \(g_1\))} \wt
	}{
		\texttt{(defrel (\(R_1\) \(\forall \vec \alpha_1\).(\(\vec x_1:\vec \tau_1\))) \(g_1\)))}, \dots, \\\\
		\texttt{(defrel (\(R_N\) \(\forall \vec \alpha_N\).(\(\vec x_N:\vec \tau_N\))) \(g_N\)))} \wt
	}
\end{mathpar}
\caption{Inference rules for polymorphic relation and program typing judgements.}
	\label{fig:poly-relation-types}
\end{figure}

As before, the \(\Gamma\) context holds relation types, and the \(\Delta\) context holds variable types.
We introduce a new \(\Delta \vdash \tau \vt\) judgement (defined in figure \ref{fig:poly-types-valid}), which asserts that the type \(\tau\) is valid under the \(\Delta\) context.
The judgements within the \(\Gamma\) context are extended to support type variables: \((R : \forall \alpha_1, \dots, \alpha_m . \tau_1,\dots,\tau_n \rightarrow)\).
As before, \(\Delta\) holds value typing judgements (of the form \(x:\tau\)). 
\(\Delta\) is also extended to hold judgements of the form \(\alpha : *\), indicating that the type variable \(\alpha\) can be treated as a valid type in the current context.
In both cases, type variables may now appear in \(\tau\).

We use \(\Delta \vdash v:\tau\) to denote the value (and therefore variable) typing judgement asserting that value \(v\) has type \(\tau\).
The value typing judgement is shown in figure \ref{fig:poly-value-types}.
The \(\Gamma;\Delta \vdash g \wt\) judgement denotes that the goal \(g\) is well-typed under the \(\Gamma\) and \(\Delta\) contexts, as shown in figure \ref{fig:poly-goal-types}.
When typechecking polymorphic relation calls, we need to verify that each callee type variable is consistently used at the same type by the caller.
For example, for a relation of type \((R:\forall \alpha . \alpha, \alpha \to)\), we should not be able to call it as \(\texttt{(\(R\) sole (left sole))}\),
where \(\alpha\) gets mapped to both \(\texttt{Unit}\) and \(\texttt{(Sum Unit Unit)}\) within a single call.
To accomplish this, we use a \emph{type variable substitution}, notated \(\sigma\), which consistently replaces type variables in the callee's type with types that are valid under the caller's context.
As long as some such \(\sigma\) exists for which the arguments typecheck correctly, the relation call itself typechecks.
The \(\Gamma \vdash Rel \wt\) judgement denotes that the relation \(Rel\) is well-typed under the \(\Gamma\) context, and the \(P \wt\) judgement denotes that the entire program is well-typed.
The relation and program typing judgements are shown in figure \ref{fig:poly-relation-types}.

In the formal semantics type variables are explicitly specified, but in practice they can be automatically extracted by looking at the argument types.
Also note that in polymorphic relation calls, type variable substitutions may replace callee type variables with caller type variables.
This is expected, as the caller itself may be a polymorphic relation, and have type variables which are valid types under \(\Delta\).
Finally, note we use the shorthand \(\alpha \in \tau\) to say ``the type variable \(\alpha\) occurs somewhere within type \(\tau\).''

\section{Monomorphizing Semantics}

Here we update the denotational semantics to handle type variables.
We use the term ``concrete types'' for types that do not contain type variables.
Here, we specifically use type variable substitutions to turn arbitrary types (which can contain type variables) into concrete types.

Temporarily ignoring substitutions, imagine we wish to operate on values of an unknown type.
To represent such a value with the array-based semantics, we need to know the size of the type so we can give it a large-enough array dimension.
Once we have space for the type, we can easily define equality and disequality operations on the type by using diagonals and off-diagonals, as before.
This generalizes to values that may only partially consist of an unknown type.
Because the specifics of the types are unknown, we have no other operations.
Thus, we can treat all unknown types of the same size equivalently.
As such, for each relation \(R\) with type variables \(\alpha, \dots, \beta\),
we generate concrete ``monomorphic'' instances of these relations \(R_{\sigma_1}, \dots, R_{\sigma_n}\), for each unique-up-to-size type variable substitution \(\sigma_1, \dots, \sigma_n\).
Note that we can easily generate types of size \(n\) as a sum type of \(n\) \(\texttt{Unit}\)s, so generating these substitutions is straightforward.

\begin{figure}
\begin{align*}
	\den{\texttt{Unit}} &= \{\texttt{sole}\} \\
	\den{\texttt{(Sum \(\tau_1\) \(\tau_2\))}} &= \{\texttt{(left \(v\))} \mid v \in \den{\tau_1}\} \cup \{\texttt{(right \(v\))} \mid v \in \den{\tau_2}\} \\
	\den{\texttt{(Prod \(\tau_1\) \(\tau_2\))}} &= \{\texttt{(pair \(v_1\) \(v_2\))} \mid v_1 \in \den{\tau_1}, v_2 \in \den{\tau_2}\} \\
\end{align*}
\caption{Denotational semantics for types, again.}
	\label{fig:poly-types-semantics}
\end{figure}

\begin{figure}
\begin{align*}
	\den{\texttt{sole}} (\delta) &= \texttt{sole} \\
	\den{\texttt{(left \(v\))}} (\delta) &= \texttt{(left \(\den{v} (\delta) \))} \\
	\den{\texttt{(right \(v\))}} (\delta) &= \texttt{(right \(\den{v} (\delta) \))} \\
	\den{\texttt{(pair \(v_1\) \(v_2\))}} (\delta) &= \texttt{(pair \(\den{v_1} (\delta) \) \(\den{v_2} (\delta)\))} \\
	\den{x} (\delta) &= \delta(x)
\end{align*}
\caption{Denotational semantics for values, again.}
	\label{fig:poly-values-semantics}
\end{figure}

\begin{figure}
\begin{align*}
	\den{\texttt{(conj \(g_1\) \(g_2\))}} (\gamma;\delta;\sigma) &= \den{g_1} (\gamma;\delta;\sigma) \times \den{g_2} (\gamma;\delta;\sigma) \\
	\den{\texttt{(disj \(g_1\) \(g_2\))}} (\gamma;\delta;\sigma) &= \den{g_1} (\gamma;\delta;\sigma) + \den{g_2} (\gamma;\delta;\sigma) \\
	\den{\texttt{(fresh ((\(x : \tau\))) \(g\))}} (\gamma;\delta;\sigma) &= \sum_{v \in \den{\sigma(\tau)}} \den{g} (\gamma;\delta,x \mapsto v;\sigma) \\
	\den{\texttt{(== \(v_1\) \(v_2\))}} (\gamma;\delta;\sigma) &=
	\begin{cases}
		1 & \textrm{if } \den{v_1} (\delta) = \den{v_2} (\delta) \\
		0 & \textrm{if } \den{v_1} (\delta) \ne \den{v_2} (\delta)
	\end{cases} \\
	\den{\texttt{(=/= \(v_1\) \(v_2\))}} (\gamma;\delta;\sigma) &=
	\begin{cases}
		1 & \textrm{if } \den{v_1} (\delta) \ne \den{v_2} (\delta) \\
		0 & \textrm{if } \den{v_1} (\delta) = \den{v_2} (\delta)
	\end{cases} \\
	\den{\texttt{(\(R\) \(\vec v\))}} (\gamma;\delta;\sigma) &= \gamma(R_{\sigma'})(\den{\vec v}(\delta))
	\textrm{ where } \vec v : \sigma'(\vec\tau) \\
	\den{\texttt{(factor \(r\))}} (\gamma;\delta;\sigma) &= r
\end{align*}
\caption{Denotational semantics for polymorphic goals.}
	\label{fig:poly-goals-semantics}
\end{figure}

\begin{figure}
\begin{gather*}
	\den{\texttt{(defrel (\(R_{\sigma}\) \(\forall \vec \alpha\) . (\(\vec x : \sigma (\vec \tau)\)) g)}} (\gamma) =
	(\vec i) \mapsto \den{g} (\gamma; \vec x \mapsto \den{\vec i};\sigma) \\
	\den{\textit{Rel}_1, \dots, \textit{Rel}_Q} = \;
	\mathrel{\textrm{fix}} \gamma \mathrel{\textrm{in}} R_1 \mapsto \den{\textit{Rel}_1} (\gamma), \dots, R_Q \mapsto \den{\textit{Rel}_Q} (\gamma)
\end{gather*}
\caption{Denotational semantics for polymorphic relations and programs.}
	\label{fig:poly-relations-semantics}
\end{figure}

As before, \(\gamma\) is a relation environment and matches the relation type context.
More specifically, \(\gamma\) is a map from relation names to multidimensional arrays.
We use the term ``concrete value'' to refer to values constructed from \texttt{sole}, \texttt{(left \(v\))}, \texttt{(right \(v\))}, and \texttt{(pair \(v_1\) \(v_2\))}; concrete values cannot contain variables.
Let \(\delta\) be a value environment with type \(\sigma(\Delta)\).
As before, \(\delta\) is a map from variable names to concrete values.
The denotation for types (figure \ref{fig:poly-types-semantics}) is still as a set of specific values.
We handle type variables by applying \(\sigma\) before taking the denotation of a type.
The denotation for values (figure \ref{fig:poly-values-semantics}) under a variable context is a specific concrete value.
The denotation for goals (figure \ref{fig:poly-goals-semantics}) is expressed as a weight for each assignment of values to variables,
but should be considered as a multidimensional array across all possible assignments.
We assume that all needed monomorphic instances of each relation are included in the full program before fixpointing,
and handle polymorphic relation calls by calling the appropriate monomorphized instance.
The semantics for relations and programs are given in figure \ref{fig:poly-relations-semantics}.

We refer to polymorphism based on type variables as \emph{parametric polymorphism}.
Historically, parametric polymorphism in functional languages disallows operations on variable-typed values other than move and copy.
While unification is roughly analogous with moving and copying, disunification does not have a direct analog.
The most similar operation might be a boolean disequality test, which is usually disallowed.
That said, we consider it fair to include disunification because we can implement it in a type-generic way similar to unification.

\section{Non-Monomorphizing Semantics}

\subsection{Equality Patterns}

In practice, semiringKanren avoids the need for monomorphization by using \emph{equality patterns} and \emph{large-enough instances} of polymorphic relations.
Let us revisit the \(\texttt{sum-swap}\) relation:
\begin{equation}
\begin{aligned}
	& \texttt{(defrel (sum-swap (\(x : \texttt{(Sum \(\alpha\) \(\beta\))}\)) (\(y : \texttt{(Sum \(\beta\) \(\alpha\))}\)))} \\
	& \quad \texttt{(disj} \\
	& \quad \quad \texttt{(fresh ((\(a : \alpha\)))} \\
	& \quad \quad \quad \texttt{(conj} \\
	& \quad \quad \quad \quad \texttt{(== \(x\) (left \(a\)))} \\
	& \quad \quad \quad \quad \texttt{(== \(y\) (right \(a\)))))} \\
	& \quad \quad \texttt{(fresh ((\(b : \beta\)))} \\
	& \quad \quad \quad \texttt{(conj} \\
	& \quad \quad \quad \quad \texttt{(== \(x\) (right \(b\)))} \\
	& \quad \quad \quad \quad \texttt{(== \(y\) (left \(b\)))))))} \\
\end{aligned}
\end{equation}
When \(|\alpha| = 3\) and \(|\beta| = 3\), \(\texttt{sum-swap}\) is denoted as follows (vertical axis is \(x\), horizontal axis is \(y\)):
\begin{equation}
	\begin{bmatrix}
		0 & 0 & 0 & 1 & 0 & 0 \\
		0 & 0 & 0 & 0 & 1 & 0 \\
		0 & 0 & 0 & 0 & 0 & 1 \\
		1 & 0 & 0 & 0 & 0 & 0 \\
		0 & 1 & 0 & 0 & 0 & 0 \\
		0 & 0 & 1 & 0 & 0 & 0
	\end{bmatrix}
\end{equation}
When we keep \(|\alpha| = 3\) but change \(|\beta| = 4\), \(\texttt{sum-swap}\) is denoted as the slightly larger array:
\begin{equation}
	\begin{bmatrix}
		0 & 0 & 0 & 0 & 1 & 0 & 0 \\
		0 & 0 & 0 & 0 & 0 & 1 & 0 \\
		0 & 0 & 0 & 0 & 0 & 0 & 1 \\
		1 & 0 & 0 & 0 & 0 & 0 & 0 \\
		0 & 1 & 0 & 0 & 0 & 0 & 0 \\
		0 & 0 & 1 & 0 & 0 & 0 & 0 \\
		0 & 0 & 0 & 1 & 0 & 0 & 0 \\
	\end{bmatrix}
\end{equation}
This clearly has a similar structure to the \(|\alpha|=|\beta|=3\) array: diagonal sub-matrices in the top right and bottom left corners, and zeros everywhere else.

Now, imagine trying to reconstruct the second relation array from the first, knowing only the type signatures.
For each entry in the second array, we consider the current value environment (which we can reconstruct from the index of the current entry),
and try to look up an entry in the first array with a matching value environment.
(For brevity, we use numbers subscripted with type size to denote the possible values here).
For example, the top-left entry of the second array has value environment \(x = \texttt{(left \(\mathtt{0_3}\))}, y = \texttt{(left \(\mathtt{0_4}\))}\).
Any entry in the first array with \(x = \texttt{(left \(\dots\))}\) and \(y = \texttt{(right \(\dots\))}\) has weight \(0\), thus we assume this entry in the second array should also have weight \(0\).
Jumping forward to the fifth entry in the first row: we know this entry has value environment \(x = \texttt{(left \(\mathtt{0_3}\))}, y = \texttt{(right \(\mathtt{0_3}\))}\)
The \(\alpha\)-typed parts of these variables both have value \(\mathtt{0_3}\); they are equal.
Similar to before, we can look up any any entry in the first array where \(x = \texttt{(left \(a\))}\), \(y = \texttt{(right \(b\))}\), and we also impose that \(a = b\).
These are the entries along the top right diagonal of the first array.
These entries all have weight \(1\), so we adopt that weight here.
Similarly, any entry in the top-right where \(a \ne b\), such as \(x = \texttt{(left \(\mathtt{1_3}\))}\) and \(y = \texttt{(right \(\mathtt{2_3}\))}\), has a weight of \(0\) that we adopt.

More generally, for each entry in the second array, we consider the \emph{equality patterns} of the value environment at that entry, and look up an entry with the same equality patterns in the first array.
When we know the non-variable-typed parts of the environment (such as the top-level \(\texttt{left}\) or \(\texttt{right}\) in the previous example),
we want to look up an entry whose non-variable-typed parts are the same.
When we know the variable-typed parts of the environment, we consider whether they are equal or disequal to other variable-typed parts,
and look up an entry that has the same equalities and disequalities.
Because in semiringKanren the only primitive relations on variable-typed values are direct equality or disequality,
considering solely the equalities/disequalities between the variable-typed parts of environments is sufficient.

Specifically, equality patterns are defined in terms of \emph{shells} and \emph{holes} of value environments.
The shell of a value environment is the non-variable-typed part, and the holes in a value environment are the variable-typed parts.
Note that holes may come from different type variables; this is important, as we can only check equality/disequality for holes of the same type variable.
We say that two value environments \(\delta_1, \delta_2\) have the same equality pattern, notated \(\delta_1 \eqpatD \delta_2\),
when they have the same shell, and when all equal pairs of holes of a type variable in \(\delta_1\) are also equal in \(\delta_2\) (and same for disequality).

For example, consider \(\delta_1 = \sigma_1(\Delta)\) and \(\delta_2 = \sigma_2(\Delta)\), with \(\Delta\) as follows:
\begin{equation}
	\Delta = x : \texttt{(Sum \(\alpha\) \(\alpha\))}, y : \alpha
\end{equation}
We will present potential equality patterns of the form \(\delta_1 \eqpatD \delta_2\), and explain why they do or do not hold.
We underline holes for illustration.
First consider:
\begin{equation}
	x \mapsto \texttt{(left \underline{sole})}, y \mapsto \texttt{\underline{sole}} \eqpatD x \mapsto \texttt{(right \underline{sole})}, y \mapsto \texttt{\underline{sole}} 
\end{equation}
Here \(\alpha = \texttt{Unit}\) for both \(\delta_1\) and \(\delta_2\).
This equality pattern does \emph{not} hold, because the shells of \(x\), \(\texttt{(left \(\dots\))}\) and \(\texttt{(right \(\dots\))}\), are not equal.
\begin{equation}
	x \mapsto \texttt{(left \(\underline{\mathtt{1_2}}\))}, y \mapsto \mathtt{\underline{0_2}} \eqpatD x \mapsto \texttt{(left \(\underline{\mathtt{1_2}}\))}, y \mapsto \mathtt{\underline{1_2}} 
\end{equation}
Here \(|\sigma(\alpha)| = 2\).
This equality pattern does \emph{not} hold, because the holes are disequal under \(\delta_1\) (\(\mathtt{1_2} \ne \mathtt{0_2}\)), but equal under \(\delta_2\) (\(\mathtt{1_2} = \mathtt{1_2}\)).
\begin{equation}
	x \mapsto \texttt{(left \(\underline{\mathtt{1_2}}\))}, y \mapsto \mathtt{\underline{0_2}} \eqpatD x \mapsto \texttt{(left \(\underline{\mathtt{0_2}}\))}, y \mapsto \mathtt{\underline{1_2}} 
\end{equation}
Here again \(|\sigma(\alpha)| = 2\).
This equality pattern \emph{does} hold. The shells are equal,
and the holes are disequal both under \(\delta_1\) (\(\mathtt{1_2} \ne \mathtt{0_2}\)) and under \(\delta_2\) (\(\mathtt{0_2} \ne \mathtt{1_2}\)).
\begin{equation}
	x \mapsto \texttt{(left \(\underline{\mathtt{0_2}}\))}, y \mapsto \mathtt{\underline{0_2}} \eqpatD x \mapsto \texttt{(left \underline{sole})}, y \mapsto \texttt{\underline{sole}} 
\end{equation}
Here \(|\sigma_1(\alpha)| = 2\) and \(\sigma_2(\alpha) = \texttt{Unit}\).
This equality pattern \emph{does} hold. The shells are equal,
and the holes are equal both under \(\delta_1\) (\(\mathtt{0_2} = \mathtt{0_2}\)) and under \(\delta_2\) (\(\texttt{sole} = \texttt{sole}\)).

\subsection{Large-Enough Relation Instances}

For \(\texttt{sum-swap}\), we have determined that we can reproduce a call with \(|\alpha| = 3, |\beta| = 4\) using an array with type variable sizes \(|\alpha| = 3, |\beta| = 3\).
Now we aim to do the same for \(\texttt{two-valued}\):
\begin{equation}
\begin{aligned}
	& \texttt{(defrel (two-valued (\(x : \alpha\)))} \\
	& \quad \texttt{(fresh ((\(y : \alpha\)))} \\
	& \quad \quad \texttt{(=/= \(x\) \(y\))))}
\end{aligned}
\end{equation}
When \(|\alpha| = 2\), \(\texttt{two-valued}\) is denoted as follows:
\begin{equation}
	\begin{bmatrix} 1 & 1 \end{bmatrix} 
\end{equation}
If we aim to use this to reproduce the array for \(\texttt{two-valued}\) when \(|\alpha| = 1\), no matter which entry we look in, we find the weight \(1\).
Thus, the denotation we calculate for \(|\alpha| = 1\) is:
\begin{equation}
	\begin{bmatrix} 1 \end{bmatrix} 
\end{equation}
which indicates success.
But we have already established \(\texttt{(two-valued sole)}\) (which has \(|\alpha| = 1\)) should fail, meaning the correct denotation should be:
\begin{equation}
	\begin{bmatrix} 0 \end{bmatrix} 
\end{equation}

To summarize, the expected behavior of \(\texttt{two-valued}\) is to fail when called with a value with type size \(1\), and succeed for any value with type size greater than \(1\).
This means that we cannot use blindly the equality pattern technique everywhere.
In particular, both the call site and the particular called relation instance must be \emph{large-enough} for equality patterns to work correctly.
Recall that the equality pattern process ``looks up'' locations in the source array where holes have the same equalities/disequalities.
In the most extreme case, all holes might be disequal.
Thus the size of each type variable needs to be ``large-enough'' to hold as many distinct values as there are possible holes.
Holes come from type variables, so we can calculate this value by finding the maximum number of times each type variable appears in an environment,
anywhere within the goal structure of the relation.

For example, in \(\texttt{sum-swap}\), \(\alpha\) and \(\beta\) both appear in the environment at most \(3\) times: twice from the relation arguments, and once under the \(\texttt{fresh}\)es.
Thus \(\texttt{sum-swap}\) is large-enough for \(|\alpha| = 3\) and \(|\beta| = 3\).
In \(\texttt{two-valued}\), \(\alpha\) appears in the environment at most twice: once from the argument, and once from the fresh.
Thus \(\texttt{two-valued}\) is large-enough when \(|\alpha| = 2\).
When a relation call is not large-enough, we do still need to generate a monomorphized version of the called relation.

\subsection{Compiling Polymorphic Programs}

While we describe the equality pattern approach in terms of ``looking up entries in the arrays,'' in practice we implement it by compiling polymorphic semiringKanren programs into non-polymorphic ones.
In particular, polymorphic relation calls are handled in two steps.
First, we call the smallest large-enough relation with fresh variables.
Then, we enforce an equality pattern between the fresh variables and the original arguments.
This is functionally similar to ``looking up'' entries: calling the relation with fresh variables brings the entries into scope, and enforcing the equality pattern makes sure we get the right weight.

Getting this to work involves some tricks.
In particular, we require semiring addition (used in \(\texttt{disj}\) and \(\texttt{fresh}\)) to be idempotent.
Intuitively, in the monomorphizing semantics, \(\texttt{fresh}\) may sum over different numbers of values, which could make goal weights depend on the size of type variables.
For example, consider the following goal:
\begin{equation}
\begin{aligned}
	& \texttt{(fresh ((\(x  : \alpha\)))} \\
	& \quad \texttt{(factor 1))}
\end{aligned}
\end{equation}
The weight from this goal is a sum \(|\alpha|\) \(1\)s, and thus may depend on the size of \(\alpha\).
This is something we are trying to avoid.
When addition is idempotent, we instead have:
\begin{equation}
	\den{\texttt{(fresh ((\(x : \alpha\))) (factor 1))}}(\dots) = 1 + \dots + 1 = 1 
\end{equation}
regardless of the size of \(\alpha\), because idempotence gives us \(1 + 1 = 1\).

\section{Proofs for Non-Monomorphizing Semantics}

We now aim to prove that this compilation scheme works as intended.
Throughout this section, we assume that we are working over a commutative semiring with multiplicative identity \(1\), additive identity/multiplicative annihilator \(0\),
and idempotent addition (so \(a + a = a\)).

\subsection{Equality Patterns Preserve Weight}

\begin{definition}
	We define a \emph{shell} using the following grammar:
	\[ \begin{array}{lrcl}
		\textrm{Shells} & s & ::= & \texttt{sole} \\
		& & | & \texttt{(left \(s\))} \\
		& & | & \texttt{(right \(s\))} \\
		& & | & \texttt{(pair \(s_1\) \(s_2\))} \\
		& & | & \texttt{(hole \(\alpha\))} \\
	\end{array} \]
\end{definition}

Shells are similar to values, but additionally have holes that are associated with a single, non-composite type variable.

Given a concrete value and a type which it inhabits after type variable substitution, we can find its shell as follows:
\begin{equation}
\begin{aligned}
	\shell_\texttt{Unit}(\texttt{sole}) &= \texttt{sole} \\
	\shell_\texttt{(Sum \(\tau_1\) \(\tau_2\))}(\texttt{(left \(v\))}) &= \texttt{(left \(\shell_{\tau_1}(v)\))} \\
	\shell_\texttt{(Sum \(\tau_1\) \(\tau_2\))}(\texttt{(right \(v\))}) &= \texttt{(right \(\shell_{\tau_2}(v)\))} \\
	\shell_\texttt{(Prod \(\tau_1\) \(\tau_2\))}(\texttt{(pair \(v_1\) \(v_2\))}) &= \texttt{(pair \(\shell_{\tau_1}(v_1)\) \(\shell_{\tau_2}(v_2)\))} \\
	\shell_\alpha(v) &= \texttt{(hole \(\alpha\))}
\end{aligned}
\end{equation}

We use \(\holes_{\alpha \in \tau}(v)\) to denote the values of the holes coming from type variable \(\alpha\) within type \(\tau\).
\begin{equation}
\begin{aligned}
	\holes_{\alpha \in \tau}(\texttt{sole}) &= \emptyset \\
	\holes_{\alpha \in \texttt{(Sum \(\tau_1\) \(\tau_2\))}}(\texttt{(left \(v\))}) &= \holes_{\alpha \in \tau_1}(v) \\
	\holes_{\alpha \in \texttt{(Sum \(\tau_1\) \(\tau_2\))}}(\texttt{(right \(v\))}) &= \holes_{\alpha \in \tau_2}(v) \\
	\holes_{\alpha \in \texttt{(Prod \(\tau_1\) \(\tau_2\))}}(\texttt{(pair \(v_1\) \(v_2\))}) &= \holes_{\alpha \in \tau_1}(v_1), \holes_{\alpha \in \tau_2}(v_2) \\
	\holes_{\alpha \in \alpha}(v) &= v \\
	\holes_{\beta \in \alpha}(v) &= \emptyset
\end{aligned}
\end{equation}

We also have versions of these functions that apply to \emph{value environments} (notated as \(\delta\) in the semantics) with pre-substitution type environments.
\begin{equation}
\begin{aligned}
	\envshell_{\emptyset}() &= \emptyset \\
	\envshell_{\Delta, \tau}(\vec x \mapsto \vec v, x \mapsto v) &= \envshell_{\Delta}(\vec x \mapsto \vec v), x \mapsto \shell_{\tau}(v)
\end{aligned}
\end{equation}
\begin{equation}
\begin{aligned}
	\envholes_{\alpha \in \emptyset}() &= \emptyset \\
	\envholes_{\Delta, \alpha \in \tau}(\vec x \mapsto \vec v, x \mapsto v) &= \envholes_{\alpha \in \Delta}(\vec v), \holes_{\alpha \in \tau}(v)
\end{aligned}
\end{equation}

Note that extracted holes are not labeled; we can find the \(\texttt{(hole \(...\))}\) where each extracted value fits by an in-order search through the shell.

\begin{lemma}
	\label{lem:values-shells-holes}
	Let \(\cdot \vdash v_1 : \sigma(\tau)\) and \(\cdot \vdash v_2 : \sigma(\tau)\) be concrete values.
	Then \(v_1 = v_2\) iff \(\shell_\tau(v_1) = \shell_\tau(v_2)\) and \(\holes_{\alpha \in \tau}(v_1) = \holes_{\alpha \in \tau}(v_2) \; \forall \alpha \in \tau\).
\end{lemma}
\begin{proof}
	Assuming \(v_1 = v_2\), by substitution we trivially have:
	\begin{equation}
		\shell_\tau(v_1) = \shell_\tau(v_2)
	\end{equation}
	\begin{equation}
		\holes_\alpha(v_1) = \holes_\alpha(v_2) \; \forall \alpha \in \tau
	\end{equation}
	Thus the \(\implies\) direction holds.

	Now assume \(\shell_\tau(v_1) = \shell_\tau(v_2)\) and \(\holes_{\alpha \in \tau}(v_1) = \holes_{\alpha \in \tau}(v_2) \; \forall \alpha \in \tau\), and aim to show \(v_1 = v_2\).
	Proof by induction on \(\cdot \vdash v_1 : \sigma(\tau)\) and \(\cdot \vdash v_2 : \sigma(\tau)\).
	\begin{itemize}
		\item Case: \(\tau = \texttt{Unit}\).
			We have \(\shell_{\texttt{Unit}}(v_1) = \shell_{\texttt{Unit}}(v_2)\).
			This is only possible when \(v_1 = v_2 = \texttt{sole}\).
		\item Case: \(\tau = \texttt{(Sum \(\tau_1\) \(\tau_2\))}\).
			We have \(\shell_{\texttt{(Sum \(\tau_1\) \(\tau_2\))}}(v_1) = \shell_{\texttt{(Sum \(\tau_1\) \(\tau_2\))}}(v_2)\).
			This is only possible in two cases.
			\begin{itemize}
				\item Subcase: \(v_1 = \texttt{(left \(v_1'\))}\) and \(v_2 = \texttt{(left \(v_2'\))}\) where \(\shell_{\tau_1}(v_1') = \shell_{\tau_1}(v_2')\).
					We know:
					\begin{equation}
						\holes_{\alpha \in \texttt{(Sum \(\tau_1\) \(\tau_2\))}}(v_1) = \holes_{\alpha \in \texttt{(Sum \(\tau_1\) \(\tau_2\))}}(v_2) \; \forall \alpha \in \tau
					\end{equation}
					so we can deduce by definition:
					\begin{equation}
						\holes_{\alpha \in \tau_1}(v_1') = \holes_{\alpha \in \tau_1}(v_2') \; \forall \alpha \in \tau
					\end{equation}
					By the inductive hypothesis, we now know \(v_1' = v_2'\).
					Thus
					\begin{equation}
						\texttt{(left \(v_1'\))} = \texttt{(left \(v_2'\))}
					\end{equation}
					and so \(v_1 = v_2\).
				\item Subcase: \(v_1 = \texttt{(right \(v_1'\))}\) and \(v_2 = \texttt{(right \(v_2'\))}\) where \(\shell_{\tau_2}(v_1') = \shell_{\tau_2}(v_2')\).
					We know:
					\begin{equation}
						\holes_{\alpha \in \texttt{(Sum \(\tau_1\) \(\tau_2\))}}(v_1) = \holes_{\alpha \in \texttt{(Sum \(\tau_1\) \(\tau_2\))}}(v_2) \; \forall \alpha \in \tau
					\end{equation}
					so we can deduce by definition:
					\begin{equation}
						\holes_{\alpha \in \tau_2}(v_1') = \holes_{\alpha \in \tau_2}(v_2') \; \forall \alpha \in \tau
					\end{equation}
					By the inductive hypothesis, we now know \(v_1' = v_2'\).
					Thus
					\begin{equation}
						\texttt{(right \(v_1'\))} = \texttt{(right \(v_2'\))}
					\end{equation}
					and so \(v_1 = v_2\).
			\end{itemize}
			In either case \(v_1 = v_2\).
		\item Case: \(\tau = \texttt{(Prod \(\tau_1\) \(\tau_2\))}\).
			We have:
			\begin{equation}
				\shell_{\texttt{(Prod \(\tau_1\) \(\tau_2\))}}(v_1) = \shell_{\texttt{(Prod \(\tau_1\) \(\tau_2\))}}(v_2)
			\end{equation}
			This is only possible when \(v_1 = \texttt{(pair \(v_1'\) \(v_1''\))}\) and \(v_2 = \texttt{(pair \(v_2'\) \(v_2''\))}\),
			where \(\shell_{\tau_1}(v_1') = \shell_{\tau_1}(v_2')\) and \(\shell_{\tau_2}(v_1'') = \shell_{\tau_2}(v_2'')\).
			We also have:
			\begin{equation}
				\holes_{\alpha \in \texttt{(Prod \(\tau_1\) \(\tau_2\))}}(v_1) = \holes_{\alpha \in \texttt{(Prod \(\tau_1\) \(\tau_2\))}}(v_2) \; \forall \alpha \in \texttt{(Prod \(\tau_1\) \(\tau_2\))}
			\end{equation}
			By definition of \(\holes\), we therefore have:
			\begin{equation}
				\holes_{\alpha \in \tau_1}(v_1'),\holes_{\alpha \in \tau_1}(v_1'') = \holes_{\alpha \in \tau_1}(v_2'),\holes_{\alpha \in \tau_1}(v_2'')
			\end{equation}
			Because \(\shell_{\tau_1}(v'_1) = \shell_{\tau_1}(v'_2)\) and \(\shell_{\tau_2}(v''_1) = \shell_{\tau_2}(v''_2)\),
			we can match up the locations of holes to deduce:
			\begin{equation}
				\holes_{\alpha \in \tau_1}(v_1') = \holes_{\alpha \in \tau_1}(v_2') \; \forall \alpha \in \tau_1
			\end{equation}
			\begin{equation}
				\holes_{\alpha \in \tau_2}(v_1'') = \holes_{\alpha \in \tau_2}(v_2'') \; \forall \alpha \in \tau_2
			\end{equation}
			By the inductive hypothesis, we have \(v_1' = v_2'\) and \(v_1'' = v_2''\).
			Thus, \(\texttt{(pair \(v_1'\) \(v_1''\))} = \texttt{(pair \(v_2'\) \(v_2''\))}\), and so \(v_1 = v_2\).
		\item Case: \(\tau = \alpha\) (for any type variable).
			By definition we know \(\holes_{\alpha \in \alpha}(v_1) = v_1\) and \(\holes_{\alpha \in \alpha}(v_2) = v_2\).
			Our second hypothesis holds for all type variables, so \(\holes_{\alpha \in \alpha}(v_1) = \holes_{\alpha \in \alpha}(v_2)\).
			Thus by transitivity of equality, \(v_1 = v_2\).
	\end{itemize}
	In any case, we have \(v_1 = v_2\).
\end{proof}

\begin{definition}
	\label{def:eqpat}
	We define \(\delta_1 \eqpatD \delta_2\) when value environments \(\delta_1\) and \(\delta_2\) have the same \emph{equality pattern} under type environment \(\Delta\).
	\begin{itemize}
		\item \(\delta_1 : \sigma_1(\Delta)\) and \(\delta_2 : \sigma_2(\Delta)\), for some type variable substitutions \(\sigma_1, \sigma_2\).
		\item \(\envshell_{\Delta}(\delta_1) = \envshell_{\Delta}(\delta_2)\), and
		\item for all type variables \(\alpha \in \Delta\), for all indices \(i, j\),
			holes \(h_1, k_1\) at those indices in \(\envholes_{\alpha \in \Delta}(\delta_1)\),
			and holes \(h_2, k_2\) at those indices in \(\envholes_{\alpha \in \Delta}(\delta_2)\),
			then \(h_1 = k_1\) and \(h_2 = k_2\), or \(h_1 \ne k_1\) and \(h_2 \ne k_2\).
	\end{itemize}
\end{definition}

Note that because the envshells for each value environment are equal, both environments must have the same number of holes in the same locations;
the third condition only makes sense if the second condition holds.
The third condition boils down to ``if pairs of holes are equal for one value environment, then they must also be equal for the other value environment'' (and similar for disequality).

\begin{lemma}
	\label{lem:eqpat-equivrel}
	\(\eqpatD\) is an equivalence relation.
\end{lemma}
\begin{proof}
	\(\eqpatD\) is defined as separate conditions on shells and holes.
	The shell condition is direct equality, so it is clearly an equivalence relation.
	It remains to show that the \(\eqpatD\) conditions for holes form an equivalence relation.

	\(\eqpatD\) is obviously reflexive; for holes \(h, k\) at type \(\alpha\), we either have \(h = k\) and \(h = k\), or \(h \ne k\) and \(h \ne k\).
	This boils down to saying that either \(h = k\) or \(h \ne k\); clearly one of these must be true.

	Let \(\delta_1 \eqpatD \delta_2\) be value environments with respective holes at indices \(i, j\): \(h_1, k_1\) and \(h_2, h_2\), for type \(\alpha\).
	We have \(h_1 = k_1\) and \(h_2 = k_2\), or \(h_1 \ne k_1\) and \(h_2 \ne k_2\).
	``And'' is symmetric, so we have \(h_2 = k_2\) and \(h_1 = k_1\), or \(h_2 \ne k_2\) and \(h_1 \ne k_1\).
	This applies to all pairs of holes at any indices, for any type variable, thus we have \(\delta_2 \eqpatD \delta_1\), and so \(\eqpatD\) is symmetric.

	Now assume \(\delta_1 \eqpatD \delta_2\) and \(\delta_2 \eqpatD \delta_3\), and aim to show \(\delta_1 \eqpatD \delta_3\).
	Let \(\alpha\) be a type variable, and choose \(\alpha\)-holes at indices \(i, j\) for each value environment: \(h_1, k_1\), \(h_2, k_2\), and \(h_3, k_3\).
	We have two cases to consider: \(h_1 = k_1\) and \(h_2 = k_2\) and \(h_3 = k_3\), or \(h_1 \ne k_1\) and \(h_2 \ne k_2\) and \(h_3 \ne k_3\).
	Any other cases do not satisfy the equality pattern conditions.
	In the first case, we have \(h_1 = k_1\) and \(h_3 = k_3\).
	In the second case, we have \(h_1 \ne k_1\) and \(h_3 \ne k_3\).
	This applies for any \(\alpha\)-holes at arbitrarily chosen indices \(i, j\) of \(\delta_1, \delta_3\), for any \(\alpha\).
	Thus we have \(d_1 \eqpatD d_3\), and so \(\eqpatD\) is transitive.

	\(\eqpatD\) is defined as separate conditions on shells and holes.
	Both conditions are reflexive, symmetric, and transitive.
	This, \(\eqpatD\) is an equivalence relation.
\end{proof}

\begin{lemma}
	\label{lem:eqpat-substitution}
	Let \(\delta_1, \delta_2\) be value environments where \(\delta_1 \eqpatD \delta_2\).
	Additionally, let \(\Delta \vdash v_1, v_2 : \tau\) be values.
	Then \(\den{v_1}(\delta_1) = \den{v_2}(\delta_1)\) if and only if \(\den{v_1}(\delta_2) = \den{v_2}(\delta_2)\).
\end{lemma}
\begin{proof}
	First assume \(\den{v_1}(\delta_1) = \den{v_2}(\delta_1)\) and aim to prove \(\den{v_1}(\delta_2) = \den{v_2}(\delta_2)\).
	Note that the converse holds by symmetry, so we only need to prove one direction.

	We know \(\shell_{\tau}(\den{v_1}(\delta_1)) = \shell_{\tau}(\den{v_2}(\delta_1))\)
	by applying the trivial part of lemma \ref{lem:values-shells-holes}.
	By assumption, \(\delta_1 \eqpatD \delta_2\), so \(\envshell_\Delta(\delta_1) = \envshell_\Delta(\delta_2)\).
	Because the envshells of \(\delta_1\) and \(\delta_2\) are equal, the shells of each variable in \(\delta_1\) and \(\delta_2\) are equal.
	Thus for any valid variable \(x : \tau'\), \(\shell_{\tau'}(\den{x}(\delta_1)) = \shell_{\tau'}(\den{x}(\delta_2))\).
	This cleanly extends to arbitrary values.
	Thus \(\shell_\tau(\den{v_1}(\delta_1)) = \shell_\tau(\den{v_1}(\delta_2))\) and \(\shell_\tau(\den{v_2}(\delta_1)) = \shell_\tau(\den{v_2}(\delta_2))\).
	By transitivity of equality, we therefore know \(\shell_\tau(\den{v_1}(\delta_2)) = \shell_{\tau}(\den{v_2}(\delta_2))\).

	For any \(\alpha \in \tau\) and valid index \(i\), we assert:
	\begin{equation}
		\holes_{\alpha \in \tau}(\den{v_1}(\delta_2))[i] = \holes_{\alpha \in \tau}(\den{v_2}(\delta_2))[i]
	\end{equation}
	Because holes come from type variables, and variable-typed values can only be constructed by by referencing variables in the environment,
	we know the values of all holes under \(\delta_2\) must come from somewhere within \(\delta_2\).
	Thus,
	\begin{equation}
		\holes_{\alpha \in \tau}(\den{v_1}(\delta_2))[i] = \envholes_{\alpha \in \Delta}(\delta_2)[j_1] 
	\end{equation}
	\begin{equation}
		\holes_{\alpha \in \tau}(\den{v_2}(\delta_2))[i] = \envholes_{\alpha \in \Delta}(\delta_2)[j_2] 
	\end{equation}
	for some \(j_1, j_2\).
	We can determine suitable \(j_1, j_2\) by figuring out which variable our current hole comes from,
	and counting how many holes occur before it in the current envshell.
	Because we have already shown that the shells are consistent across \(v_1, v_2\) and \(\delta_1, \delta_2\),
	we can look up the values of the same holes under \(\delta_1\):
	\begin{equation}
		\holes_{\alpha \in \tau}(\den{v_1}(\delta_1))[i] = \envholes_{\alpha \in \Delta}(\delta_1)[j_1] 
	\end{equation}
	\begin{equation}
		\holes_{\alpha \in \tau}(\den{v_2}(\delta_1))[i] = \envholes_{\alpha \in \Delta}(\delta_1)[j_2] 
	\end{equation}
	Because \(\den{v_1}(\delta_1) = \den{v_2}(\delta_2)\) by assumption, we therefore have:
	\begin{equation}
		\holes_{\alpha \in \tau}(\den{v_1}(\delta_1))[i] = \holes_{\alpha \in \tau}(\den{v_2}(\delta_1))[i] 
	\end{equation}
	and so:
	\begin{equation}
		\envholes_{\alpha \in \Delta}(\delta_1)[j_1] = \envholes_{\alpha \in \Delta}(\delta_1)[j_2] 
	\end{equation}
	by transitivity of equality.
	Because of this and \(\delta_1 \eqpatD \delta_2\), we therefore must also have:
	\begin{equation}
		\envholes_{\alpha \in \Delta}(\delta_2)[j_1] = \envholes_{\alpha \in \Delta}(\delta_2)[j_2] 
	\end{equation}
	by definition of \(\eqpatD\).
	Therefore by transitivity of equality:
	\begin{equation}
		\holes_{\alpha \in \tau}(\den{v_1}(\delta_2))[i] = \holes_{\alpha \in \tau}(\den{v_2}(\delta_2))[i] 
	\end{equation}
	This holds for any \(\alpha \in \tau\) and any valid index \(i\), so we have:
	\begin{equation}
		\holes_{\alpha \in \Delta}(\den{v_1}(\delta_2)) = \holes_{\alpha \in \Delta}(\den{v_2}(\delta_2)) \; \forall \alpha \in \tau
	\end{equation}

	We have shown \(\shell_\tau(\den{v_1}(\delta_2)) = \shell_{\tau}(\den{v_2}(\delta_2))\) and
	\(\holes_{\alpha \in \tau}(\den{v_1}(\delta_2)) = \holes_{\alpha \in \tau}(\den{v_2}(\delta_2) \; \forall \alpha \in \tau\).
	Thus by lemma \ref{lem:values-shells-holes}, \(\den{v_1}(\delta_2) = \den{v_2}(\delta_2)\).
\end{proof}

\begin{definition}
	We define the \emph{(maximum) number of occurrences of type variable \(\alpha\)} in type \(\tau\), denoted \(\#_\alpha \tau\), as follows:
	\begin{itemize}
		\item \(\#_\alpha \texttt{Unit} = 0\)
		\item \(\#_\alpha \texttt{(Sum \(\tau_1\) \(\tau_2\))} = \max(\#_\alpha \tau_1, \#_\alpha \tau_2)\)
		\item \(\#_\alpha \texttt{(Prod \(\tau_1\) \(\tau_2\))} = \#_\alpha \tau_1 + \#_\alpha \tau_2\)
		\item \(\#_\alpha \alpha = 1\)
	\end{itemize}
\end{definition}

More specifically, this is the maximum number of \(\alpha\)-typed holes that can occur in a value of type \(\tau\).
This generalizes cleanly to type environments:

\begin{definition}
	We define the \emph{(maximum) number of occurrences of type variable \(\alpha\)} in type environment \(\Delta\), denoted \(\#_\alpha \Delta\), as follows:
	\begin{itemize}
		\item \(\#_\alpha \emptyset = 0\)
		\item \(\#_\alpha \tau, \Delta = \#_\alpha \tau + \#_\alpha \Delta\)
	\end{itemize}
\end{definition}

\begin{definition}
	We define the \emph{(maximum) number of occurrences of type variable \(\alpha\) in goal \(g\) under type environment \(\Delta\)}, denoted \(\#_\alpha g(\Delta)\), inductively:
	\begin{itemize}
		\item \(\#_\alpha \texttt{(conj \(g_1\) \(g_2\))}(\Delta) = \max(\#_\alpha g_1, \#_\alpha g_2)\)
		\item \(\#_\alpha \texttt{(disj \(g_1\) \(g_2\))}(\Delta) = \max(\#_\alpha g_1, \#_\alpha g_2)\)
		\item \(\#_\alpha \texttt{(fresh ((\(x : \tau\))) \(g\))}(\Delta) = \#_\alpha g(\Delta,x:\tau)\)
		\item \(\#_\alpha \texttt{(== \(v_1\) \(v_2\))}(\Delta) = \#_\alpha \Delta\)
		\item \(\#_\alpha \texttt{(=/= \(v_1\) \(v_2\))}(\Delta) = \#_\alpha \Delta\)
		\item \(\#_\alpha \texttt{(\(R\) \(\vec v\))}(\Delta) = \#_\alpha \Delta\)
		\item \(\#_\alpha \texttt{(factor \(r\))}(\Delta) = \#_\alpha \Delta\)
	\end{itemize}
\end{definition}

This makes it easy to count type variable occurrences for relations:

\begin{definition}
	We define the \emph{(maximum) number of occurrences of type variable \(\alpha\) in relation \(R\)}, denoted \(\#_\alpha R\),
	as the (maximum) of occurrences of \(\alpha\) in the top-level goal, with the type environment determined from the arguments.
	\begin{align*}
		\#_\alpha \texttt{(defrel (\(R\) \(\forall \alpha \dots\) (\(x_1:\tau_1\)) \(\dots\) (\(x_n : \tau_1\))) \(g\))} = \#_\alpha g(x_1:\tau_1,\dots,x_n:\tau_n)
	\end{align*}
\end{definition}

\begin{lemma}
	\label{lem:eqpat-extend}
	Let \(\Delta \vdash v_1 : \sigma_1(\tau)\).
	If \(|\sigma_2(\alpha)| \ge \#_\alpha \Delta, x:\tau \; \forall \alpha \in \Delta\),
	and \(\delta_1 \eqpatD \delta_2\),
	then there exists \(\Delta \vdash v_2 : \sigma_2(\tau)\) such that \(\delta_1, x \mapsto v_1 \eqpat_{\Delta, x:\tau} \delta_2, x \mapsto v_2\).
\end{lemma}
\begin{proof}
	For the equality pattern to continue to hold, we must have \(\shell_\tau(v_1) = \shell_\tau(v_2)\).
	It remains to find appropriate values for the holes of \(v_2\).

	Let \(\alpha \in \tau\), and let \(i\) be the smallest index such that \(\envholes_{\alpha \in \Delta, \tau}(\delta_2)[i]\) does not yet have an assigned value.
	For \(\delta_1\), try to find index \(j < i\) such that \(\envholes_{\alpha \in \Delta, x:\tau}(\delta_1)[i] = \envholes_{\alpha \in \Delta, x:\tau}(\delta_1)[j]\).
	If such \(j\) exists, assign the value of \(\envholes_{\alpha \in \Delta, x:\tau}(\delta_2)[j]\) to \(\envholes_{\alpha \in \Delta, x:\tau}(\delta_2)[i]\).
	If such \(j\) does not exist, choose an \(\sigma_2(\tau)\)-typed value that does not yet occur in \(\envholes_{\alpha \in \Delta, x:\tau}(\delta_2)\),
	and assign it to \(\envholes_{\alpha \in \Delta, \tau}(\delta_2)[i]\).
	Note that it is always possible to find such a value: \(|\sigma_2(\alpha)| \ge \#_\alpha \Delta,x:\tau\),
	so \(\sigma_2(\alpha)\) has at least as many distinct values as there are \(\alpha\)-typed holes.
\end{proof}

\begin{definition}
	Given the context \(\gamma : \Gamma\) and some substitution \(\sigma\), we say a monomorphic instance \(R_{\sigma_1}\) of relation \(R\) is \emph{large-enough} if
	for all \(\Delta \vdash \vec v : \vec \tau\) satisfying \(\Gamma;\Delta \vdash \texttt{(\(R\) \(\vec v\))} \wt\),
	and for all monomorphic instances \(R_{\sigma_2}\) of \(R\) with \(|\sigma_2(\alpha)| \ge |\sigma_1(\alpha)| \; \forall \alpha \in \vec \tau\),
	and for all \(\delta_1 : \sigma_1(\Delta)\), \(\delta_2 : \sigma_2(\Delta)\) with \(\delta_1 \eqpatD \delta_2\),
	the following holds:
	\[\den{\texttt{(\(R_{\sigma_1}\) \(\vec v\))}}(\gamma;\delta_1;\sigma) = \den{\texttt{(\(R_{\sigma_2}\) \(\vec v\))}}(\gamma;\delta_2;\sigma)\]
\end{definition}

\begin{theorem}
	\label{thm:eqpat-weight}
	Let \(\Gamma; \Delta \vdash g \wt\), \(\delta_1 : \sigma_1(\Delta)\), \(\delta_2 : \sigma_2(\Delta)\),
	and all relations in \(\gamma : \Gamma\) be large-enough.
	If \(\delta_1 \eqpatD \delta_2\),
	\(|\sigma_1(\alpha)| \ge \#_\alpha g(\Delta) \; \forall \alpha \in \Delta\),
	\(|\sigma_2(\alpha)| \ge \#_\alpha g(\Delta) \; \forall \alpha \in \Delta\),
	and \(+\) is idempotent (\(a + a = a\)), then
	\(\den{g}(\gamma; \delta_1; \sigma_1) = \den{g}(\gamma; \delta_2; \sigma_2)\).
\end{theorem}
\begin{proof}
	Proof by induction on the goal structure of \(g\).
	\begin{itemize}
		\item \(g = \texttt{(conj \(g_1\) \(g_2\))}\).
			By the inductive hypothesis:
			\begin{equation}
				\den{g_1}(\gamma; \delta_1; \sigma_1) = \den{g_1}(\gamma; \delta_2; \sigma_2)
			\end{equation}
			\begin{equation}
				\den{g_2}(\gamma; \delta_2; \sigma_2) = \den{g_2}(\gamma; \delta_2; \sigma_2)
			\end{equation}
			Thus:
			\begin{equation}
				\den{g_1}(\gamma; \delta_1; \sigma_1) * \den{g_2}(\gamma; \delta_1; \sigma_1) = \den{g_1}(\gamma; \delta_2; \sigma_2) * \den{g_2}(\gamma; \delta_2; \sigma_2)\
			\end{equation}
			and so \(\den{g}(\gamma; \delta_1; \sigma_1) = \den{g}(\gamma; \delta_2; \sigma_2)\).
		\item \(g = \texttt{(disj \(g_1\) \(g_2\))}\).
			By the inductive hypothesis:
			\begin{equation}
				\den{g_1}(\gamma; \delta_1; \sigma_1) = \den{g_1}(\gamma; \delta_2; \sigma_2)
			\end{equation}
			\begin{equation}
				\den{g_2}(\gamma; \delta_2; \sigma_2) = \den{g_2}(\gamma; \delta_2; \sigma_2)
			\end{equation}
			Thus:
			\begin{equation}
				\den{g_1}(\gamma; \delta_1; \sigma_1) + \den{g_2}(\gamma; \delta_1; \sigma_1) = \den{g_1}(\gamma; \delta_2; \sigma_2) + \den{g_2}(\gamma; \delta_2; \sigma_2)
			\end{equation}
			and so \(\den{g}(\gamma; \delta_1; \sigma_1) = \den{g}(\gamma; \delta_2; \sigma_2)\).
		\item \(g = \texttt{(fresh ((\(x : \tau\))) \(g'\))}\).
			We know: \(|\sigma_1(\alpha)| \ge \#_\alpha g(\Delta) \; \forall \alpha \in \Delta\),
			so \(|\sigma_1(\alpha)| \ge \#_\alpha g'(\Delta,x:\tau) \; \forall \alpha \in \Delta\).
			By lemma \ref{lem:eqpat-extend}, for any \(\sigma_1(\Delta) \vdash v_1 : \sigma_1(\tau)\),
			there exists \(\sigma_2(\Delta) \vdash v_2 : \sigma_2(\tau)\) where \(\delta_1, x \mapsto v_1 \eqpat_{\Delta, x:\tau} \delta_2, x \mapsto v_2\).
			Similarly, \(|\sigma_2(\alpha)| \ge \#_\alpha g'(\Delta,x:\tau) \; \forall \alpha \in \Delta\), so for any \(\sigma_2(\Delta) \vdash v_2 : \sigma_2(\tau)\),
			there exists \(\sigma_1(\Delta) \vdash v_1 : \sigma_1(\tau)\) where \(\delta_1, x \mapsto v_1 \eqpat_{\Delta, x:\tau} \delta_2, x \mapsto v_2\).
			Thus:
			\begin{equation}
				\den{\texttt{(fresh ((\(x : \tau\))) \(g'\))}}(\gamma; \delta_1; \sigma_1) = w_1 + \dots + w_n
			\end{equation}
			\begin{equation}
				\den{\texttt{(fresh ((\(x : \tau\))) \(g'\))}}(\gamma; \delta_2; \sigma_2) = w'_1 + \dots + w'_m
			\end{equation}
			where each weight \(w_1, \dots, w_n\) appears at least once in \(w'_1, \dots, w'_m\),
			and each weight \(w'_1, \dots, w'_m\) appears at least once in \(w_1, \dots, w_n\).
			Let \(r_1, \dots, r_k\) be the unique weights appearing in both sums.
			We require the semiring to be commutative, so we can group identical weights together.
			We therefore have:
			\begin{equation}
				\den{\texttt{(fresh ((\(x : \tau\))) \(g'\))}}(\gamma; \delta_1; \sigma_1) = (r_1 + \dots) + \dots + (r_k + \dots)
			\end{equation}
			\begin{equation}
				\den{\texttt{(fresh ((\(x : \tau\))) \(g'\))}}(\gamma; \delta_2; \sigma_2) = (r_1 + \dots\dots) + \dots + (r_k + \dots\dots)
			\end{equation}
			We require \(+\) to be idempotent.
			Thus:
			\begin{equation}
				\den{\texttt{(fresh ((\(x : \tau\))) \(g'\))}}(\gamma; \delta_1; \sigma_1) = r_1 + \dots + r_k
			\end{equation}
			\begin{equation}
				\den{\texttt{(fresh ((\(x : \tau\))) \(g'\))}}(\gamma; \delta_2; \sigma_2) = r_1 + \dots + r_k
			\end{equation}
			\begin{equation}
				\den{\texttt{(fresh ((\(x : \tau\))) \(g'\))}}(\gamma; \delta_1; \sigma_1) = \den{\texttt{(fresh ((\(x : \tau\))) \(g'\))}}(\gamma; \delta_2; \sigma_2)
			\end{equation}
		\item \(g = \texttt{(== \(v_1\) \(v_2\))}\).
			By lemma \ref{lem:eqpat-substitution}, we either have:
			\begin{equation}
				\den{v_1}(\delta_1) = \den{v_2}(\delta_1)
			\end{equation}
			\begin{equation}
				\den{v_1}(\delta_2) = \den{v_2}(\delta_2)
			\end{equation}
			or
			\begin{equation}
				\den{v_1}(\delta_1) \ne \den{v_2}(\delta_1)
			\end{equation}
			\begin{equation}
				\den{v_1}(\delta_2) \ne \den{v_2}(\delta_2)
			\end{equation}
			In the first case:
			\begin{equation}
				\den{\texttt{(== \(v_1\) \(v_2\))}}(\gamma; \delta_1; \sigma_1) = 1
			\end{equation}
			\begin{equation}
				\den{\texttt{(== \(v_1\) \(v_2\))}}(\gamma; \delta_2; \sigma_2) = 1
			\end{equation}
			In the second case:
			\begin{equation}
				\den{\texttt{(== \(v_1\) \(v_2\))}}(\gamma; \delta_1; \sigma_1) = 0
			\end{equation}
			\begin{equation}
				\den{\texttt{(== \(v_1\) \(v_2\))}}(\gamma; \delta_2; \sigma_2) = 0
			\end{equation}
			In either case, \(\den{g}(\gamma; \delta_1; \sigma_1) = \den{g}(\gamma; \delta_2; \sigma_2)\).
		\item \(g = \texttt{(=/= \(v_1\) \(v_2\))}\).
			By lemma \ref{lem:eqpat-substitution}, we either have
			\begin{equation}
				\den{v_1}(\delta_1) = \den{v_2}(\delta_1)
			\end{equation}
			\begin{equation}
				\den{v_1}(\delta_2) = \den{v_2}(\delta_2)
			\end{equation}
			or
			\begin{equation}
				\den{v_1}(\delta_1) \ne \den{v_2}(\delta_1)
			\end{equation}
			\begin{equation}
				\den{v_1}(\delta_2) \ne \den{v_2}(\delta_2)
			\end{equation}
			In the first case:
			\begin{equation}
				\den{\texttt{(=/= \(v_1\) \(v_2\))}}(\gamma; \delta_1; \sigma_1) = 0
			\end{equation}
			\begin{equation}
				\den{\texttt{(== \(v_1\) \(v_2\))}}(\gamma; \delta_2; \sigma_2) = 0
			\end{equation}
			In the second case:
			\begin{equation}
				\den{\texttt{(=/= \(v_1\) \(v_2\))}}(\gamma; \delta_1; \sigma_1) = 1
			\end{equation}
			\begin{equation}
				\den{\texttt{(== \(v_1\) \(v_2\))}}(\gamma; \delta_2; \sigma_2) = 1
			\end{equation}
			In either case, \(\den{g}(\gamma; \delta_1; \sigma_1) = \den{g}(\gamma; \delta_2; \sigma_2)\).
		\item \(g = \texttt{(\(R\) \(\vec v\))}\).
			\begin{equation}
				\den{\texttt{(\(R\) \(\vec v\))}}(\gamma;\delta_1;\sigma_1) = \gamma(R_{\sigma'})(\den{\vec v}(\delta_1))
			\end{equation}
			\begin{equation}
				\den{\texttt{(\(R\) \(\vec v\))}}(\gamma;\delta_2;\sigma_2) = \gamma(R_{\sigma''})(\den{\vec v}(\delta_2))
			\end{equation}
			By the hypothesis that all relations in \(\gamma\) are large-enough, and that \(\delta_1 \eqpatD \delta_2\), these are equivalent.
		\item \(g = \texttt{(factor \(k\))}\).
			\begin{equation}
				\den{\texttt{(factor \(k\))}}(\gamma; \delta_1; \sigma_1) = k = \den{\texttt{(factor \(k\))}}(\gamma; \delta_2; \sigma_2)
			\end{equation}
	\end{itemize}
\end{proof}

The large-enough condition on \(\gamma\) is restrictive.
For example, given \(\Delta \vdash x : \alpha\) with \(\sigma_1(\alpha) = 1\) and \(\sigma_2(\alpha) = 2\),
then \(\gamma\) is not allowed to hold relations like \(\texttt{two-valued}\).
Here, \(\texttt{(two-valued \(x\))}\) would fail under \(\delta_1\) but succeed under \(\delta_2\).
Ideally, we should use something smarter than \(\#_\alpha g(\Delta)\), so type variable sizes account for these sorts of relation calls.
A smarter variant may have relation calls push size requirements upward to their arguments.
Also note that within the fixpoint process, all relations in \(\gamma\) start out by being large-enough: they consistently fail in all entries.
This theorem lets us conclude that each fixpoint iteration preserves arrays in \(\gamma\) being large-enough,
as long as the goal definitions of the relations only include large-enough relation calls.

\subsection{Compiling Polymorphic Relation Calls}

We now aim to compile arbitrary large-enough relation calls to instead call the smallest large-enough relation instance.
While it remains necessary to keep monomorphized relation instances for calls smaller than the smallest large-enough instance.
this allows us to remove all monomorphized instances larger than the smallest large-enough instance.
In practice, implementation of this compilation scheme is verbose, so we reason here with an idealized version of the process.
Rather than a formal demonstration of correctness, the following proofs aim to sketch an intuitive idea of why this method works.

\begin{lemma}
	\label{lem:no-factor-weight}
	Let \(\Gamma; \Delta \vdash g \wt\).
	If \(g\) does include not include \(\mathtt{factor}\) or relation calls, then for any valid \(\gamma, \delta, \sigma\), \(\den{g}(\gamma;\delta;\sigma)\) either has weight \(0\) or \(1\).
\end{lemma}
\begin{proof}
	Proof by induction on \(g\).
	\begin{itemize}
		\item \(g = \texttt{(conj \(g_1\) \(g_2\))}\).
			By the inductive hypothesis, \(\den{g_1}(\gamma;\delta;\sigma) = 0 \text{ or } 1\), and \(\den{g_2}(\gamma;\delta;\sigma) = 0 \text{ or } 1\).
			Thus the weight of \(\den{g}(\gamma;\delta;\sigma)\) is one of: \(0 * 0\), \(0 * 1\), \(1 * 0\), and \(1 * 1\).
			Because \(0\) is the multiplicative annihlator and \(1\) is the multiplicative identity, the first three cases are equal to \(0\) and the last case is equal to \(1\).
		\item \(g = \texttt{(disj \(g_1\) \(g_2\))}\).
			By the inductive hypothesis, \(\den{g_1}(\gamma;\delta;\sigma) = 0 \text{ or } 1\), and \(\den{g_2}(\gamma;\delta;\sigma) = 0 \text{ or } 1\).
			Thus the weight of \(\den{g}(\gamma;\delta;\sigma)\) is one of: \(0 + 0\), \(0 + 1\), \(1 + 0\), and \(1 + 1\).
			Because \(0\) is the additive identity, the first case is equal to \(0\), and the second and third cases are equal to \(1\).
			Because we require \(+\) to be idempotent, the fourth case is equal to \(1\).
		\item \(g = \texttt{(fresh ((\(x : \tau\))) \(g'\))}\).
			Because the inductive hypothesis applies for any valid \(\delta\), we know \(\den{g'}(\gamma;\delta, x \mapsto v;\sigma) = 0 \text{ or } 1\) for all \(v\).
			Our semiring is commutative, so we can rearrange the sum denoted by the \(\texttt{fresh}\) as follows:
			\begin{equation}
				(0 + \dots) + (1 + \dots)
			\end{equation}
			We know the sum must contain at least one summand (because the smallest semiringKanren type is \(\texttt{Unit}\), so we at least have \(v = \texttt{sole}\)),
			but we could potentially have either no \(0\)s or no \(1\)s in the sum.
			If there are no \(1\)s in the sum, then we have:
			\begin{equation}
				0 + \dots = 0
			\end{equation}
			If we have at least one \(1\) in the sum, then we have:
			\begin{equation}
				(0 + \dots) + (1 + \dots) = 1
			\end{equation}
			This holds because we require addition to be idempotent (so \(1 + 1 = 1\)), and any \(0\)s (if they occur) are the additive identity.
		\item \(g = \texttt{(== \(v_1\) \(v_2\))}\).
			We know \(\den{g}(\gamma;\delta;\sigma) = 0 \text{ or } 1\) by defininition in the denotational semantics.
		\item \(g = \texttt{(=/= \(v_1\) \(v_2\))}\).
			We know \(\den{g}(\gamma;\delta;\sigma) = 0 \text{ or } 1\) by defininition in the denotational semantics.
		\item \(g = \texttt{(\(R\) \(v\) \(\dots\))}\).
			By hypothesis, this case does not occur.
		\item \(g = \texttt{(factor \(r\))}\).
			By hypothesis, this case does not occur.
	\end{itemize}
\end{proof}

We now define a function to generate code to enforce equality patterns between two sub-environments.
Because the generated code should work with potentially any value sub-environments, we define the function in terms of the type sub-environments.
Let \(\Delta_1 = \sigma_1(\Delta)\), \(\Delta_2 = \sigma_2(\Delta)\), and \(\alpha, \dots \in \Delta\).
For convenience, assume \(\envshell\) and \(\envholes\) operate on type sub-environments like \(\Delta_1\), \(\Delta_2\)
by extracting shells or all possible holes as values, written in terms of ancillary variables.
In practice, this is handled manually.
\begin{equation}
\begin{aligned}
	& \enforceeqpatD(\Delta_1; \Delta_2) = \\
	& \quad \texttt{(fresh (\textrm{ancillary variables for \(\envshell\) and \(\envholes\)})} \\
	& \quad \quad \texttt{(conj }\\
	& \quad \quad \quad \textrm{deconstruct \(\Delta_1\) and \(\Delta_2\) into the ancillary variables} \\
	& \quad \quad \quad \texttt{(== \(\envshell_\Delta(\Delta_1)\) \(\envshell_{\Delta}(\Delta_2)\))} \\
	& \quad \quad \quad \textrm{for \(i,j \in 1,\dots,\#_\alpha\Delta\)} \\
	& \quad \quad \quad \quad \texttt{(disj} \\
	& \quad \quad \quad \quad \quad \texttt{(conj} \\
	& \quad \quad \quad \quad \quad \quad \texttt{(== \(\envholes_{\alpha \in \Delta}(\Delta_1)[i]\) \(\envholes_{\alpha \in \Delta}(\Delta_1)[j]\))} \\
	& \quad \quad \quad \quad \quad \quad \texttt{(== \(\envholes_{\alpha \in \Delta}(\Delta_2)[i]\) \(\envholes_{\alpha \in \Delta}(\Delta_2)[j]\)))} \\
	& \quad \quad \quad \quad \quad \texttt{(conj} \\
	& \quad \quad \quad \quad \quad \quad \texttt{(=/= \(\envholes_{\alpha \in \Delta}(\Delta_1)[i]\) \(\envholes_{\alpha \in \Delta}(\Delta_1)[j]\))} \\
	& \quad \quad \quad \quad \quad \quad \texttt{(=/= \(\envholes_{\alpha \in \Delta}(\Delta_2)[i]\) \(\envholes_{\alpha \in \Delta}(\Delta_2)[j]\))))} \\
	& \quad \quad \quad \dots \texttt{))}
\end{aligned}
\end{equation}

Once we have ancillary variables to work with, the structure of the generated code directly mirrors the definition of \(\eqpatD\).
First, we check that the \(\envshell\)s of the two sub-environments are equal.
Then for each pair of holes, we check that either both pairs of holes are equal, or both pairs of holes are disequal.

First note that deconstructing \(\Delta_1\) and \(\Delta_2\) into ancillary variables does not succeed
when the ancillary variable values in the current environment do not match up with the current values inhabiting \(\Delta_1\) and \(\Delta_2\).
When this happens, the whole \(\texttt{conj}\) fails, and so no weight is added into the summation denoted by the outer \(\texttt{fresh}\).
This means that we can ignore these failing cases in our reasoning, and only consider cases where deconstructing \(\Delta_1\) and \(\Delta_2\) succeeds.
The deconstruction process does not need to use \(\texttt{factor}\), so by lemma \ref{lem:no-factor-weight}, it either has weight \(0\) or \(1\).
Thus successes of the deconstruction process have weight \(1\), which is the multiplicative identity, so there is no effect on the weight of the top-level \(\texttt{conj}\).
Overall, we can ignore this subgoal in our reasoning going forward, other than assuming that the environment has been properly deconstructed into ancillary variables.

Similarly note that with how \(\envholes\) operates on type environments, it may process holes that are not valid given the current \(\envshell\) structure.
For example, consider the type environment
\begin{equation}
	\Delta = x : \texttt{(Sum \(\alpha\) \(\alpha\))}
\end{equation}
where \(\sigma_1(\alpha) = \texttt{Unit}\).
Then if \(\envshell_{\alpha \in \Delta}(\Delta_1) = x \mapsto \texttt{(left (hole \(\alpha\)))}\), the first \(\alpha\)-hole (at index \(0\) in \(\envholes_{\alpha \in \Delta}(\Delta_1)\)) is valid,
whereas the second \(\alpha\)-hole (at index \(1\) in \(\envholes_{\alpha \in \Delta}(\Delta_1)\)) is not.
This second \(\alpha\)-hole only exists in the shell \(\texttt{(right (hole \(\alpha\)))}\), which is not the current shell.

In practice, this is not an issue.
It simply means that during the deconstruction process, no constraints are imposed on the ancillary variables corresponding to these holes.
While the ``equality patterning'' part of the generated code may fail for certain assignments of these ancillary variables,
it is guaranteed to succeed for at least one assignment in the outer \(\texttt{fresh}\).
This is because for any two variables of the same type, there must have at least one value for which they are equal.
Nowhere in the generated code causes a failure for these ``unconstrained'' ancillary variables,
so at least the first subcase of the equality patterning \(\texttt{disj}\) is guaranteed to succeed for at least one assignment.
The \(\texttt{fresh}\) sums over all such assignments; our semiring is idempotent and commutative, so if there are multiple such successes, the overall weight is not effected.

\begin{lemma}
	\label{lem:enforce_eqpat}
	Let \(\delta_1 : \Delta_1 = \sigma_1(\Delta)\) and \(\delta_2 : \Delta_2 = \sigma_2(\Delta)\).
	Then for any relation environment \(\gamma\) and type variable substitution \(\sigma\),
	either:
	\[
		\den{\enforceeqpatD(\Delta_1;\Delta_2)}(\gamma;\delta_1,\delta_2;\sigma) = 1
	\]
	and \(\delta_1 \eqpatD \delta_2\),
	or
	\[
		\den{\enforceeqpatD(\Delta_1;\Delta_2)}(\gamma;\delta_1,\delta_2;\sigma) = 0
	\]
	and \(\delta_1 \not\eqpatD \delta_2\).
\end{lemma}
\begin{proof}
	We either have \(\delta_1 \eqpatD \delta_2\), or \(\delta_1 \not\eqpatD \delta_2\).

	First assume \(\delta_1 \eqpatD \delta_2\) and aim to show:
	\begin{equation}
		\den{\enforceeqpatD(\Delta_1;\Delta_2)}(\gamma;\delta_1,\delta_2;\sigma) = 0
	\end{equation}
	\(\delta_1 \eqpatD \delta_2\), so \(\envshell_\Delta(\delta_1) = \envshell_\Delta(\delta_2)\),
	so
	\begin{equation}
		\texttt{(== \(\envshell_{\Delta}(\Delta_1)\) \(\envshell_{\Delta}(\Delta_2)\))}
	\end{equation}
	succeeds under \(\delta_1, \delta_2\).
	Similarly, for any \(\alpha\), for any holes \(h_1, k_1\) at indices \(i, j\) in \(\envholes_{\alpha \in \Delta}(\delta_1)\),
	and holes \(h_2, k_2\) at the same indices \(i, j\) in \(\envholes_{\alpha \in \Delta}(\delta_2)\),
	then either \(h_1 = k_1\) and \(h_2 = k_2\), or \(h_1 \ne k_1\) or \(h_2 \ne k_2\).
	Thus exactly one branch of the \(\texttt{disj}\) succeeds for each \(\alpha, i, j\) (and we can ignore holes that do not ``exist'' here), so all such \(\texttt{disj}\)s have weight 1.
	We have also argued that deconstructing \(\Delta_1, \Delta_2\) into ancillary variables succeeds, and relatedly that we can ignore the outer \(\texttt{fresh}\).
	Thus we have a conjunction of subgoals of weight 1, which is the multiplicative identity, so the total weight is 1.

	Now assume \(\delta_1 \not\eqpatD \delta_2\), and aim to show:
	\begin{equation}
		\den{\enforceeqpatD(\Delta_1;\Delta_2)}(\gamma;\delta_1,\delta_2;\sigma) = 1
	\end{equation}
	We know \(\Delta_1 = \sigma_1(\Delta)\) and \(\Delta_2 = \sigma_2(\Delta)\).
	We have \(\delta_1 \eqpatD \delta_2\), so either \(\envshell_\Delta(\delta_1) \ne \envshell_\Delta(\delta_2)\),
	or for some \(\alpha\) and holes \(h_1, k_1\) at indices \(i, j\) in \(\envholes_{\alpha \in \Delta}(\delta_1)\),
	and holes \(h_2, k_2\) at the same indices \(i, j\) in \(\envholes_{\alpha \in \Delta}(\delta_2)\),
	we have \(h_1 = k_1\) and \(h_2 \ne k_2\), or \(h_1 \ne k_1\) and \(h_2 = k_2\).
	
	In the first case, \(\envshell_\Delta(\delta_1) \ne \envshell_\Delta(\delta_2)\).
	Thus the goal
	\begin{equation}
		\texttt{(== \(\envshell_\Delta(\Delta_1)\) \(\envshell_\Delta(\Delta_2)\))}
	\end{equation}
	fails under \(\delta_1, \delta_2\).
	This is denoted by weight \(0\), which is a multiplicative annihilator.
	Thus the outer \(\texttt{conj}\) has weight \(0\), and so the generated code has weight \(0\) under \(\delta_1, \delta_2\).

	In the second case, we either have \(h_1 = k_1\) and \(h_2 \ne k_2\), or \(h_1 \ne k_1\) and \(h_2 \ne k_2\),
	for holes \(h_1, k_1\) at indices \(i, j\) in \(\envholes_{\alpha \in \Delta}(\Delta_1)\) and holes \(h_2, k_2\) at the same indices \(i, j\) in \(\envholes_{\alpha \in \Delta}(\Delta_2)\).
	Thus, both subcases of the \(\texttt{disj}\) fail, and so the \(\texttt{disj}\) itself fails, returning weight \(0\).
	As before, this forces the weight of the outer \(\texttt{conj}\), and so weight of the generated code as a whole, to be \(0\).

	In either case, if \(\delta_1 \not\eqpatD \delta_2\), then \(\den{\enforceeqpatD(\Delta_1; \Delta_2)}(\gamma;\delta_1,\delta_2;\sigma) = 0\).
\end{proof}

Note that \(\gamma\) and \(\sigma\) are never applied in the generated code, so they have no effect on the final weight.

We can now compile relation calls.
Within context \(\gamma : \Gamma\), suppose that we want to compile large-enough  \(\texttt{(\(R_{\sigma_1}\) \(\vec v\))}\). 
Let \(R_{\sigma_2}\) be the smallest large-enough monomorphic instance of \(R\).
When \(\Delta_1 = \sigma_1(\Delta)\), \(\Delta_2 = \sigma_2(\Delta)\),
we use the notation \(\vec{v}[\Delta_1 \mapsto \Delta_2]\) for replacing variables from sub-environment \(\Delta_1\) with variables from sub-environment \(\Delta_2\).
Given \(\Gamma; \Delta \vdash \texttt{(\(R\) \(\vec v\))} \wt\), we compile as follows:
\begin{equation}
\begin{aligned}
	& \compile(\texttt{(\(R_{\sigma_1}\) \(\vec v\))}) = \\
	& \quad \texttt{(fresh ((\(\delta_2 : \Delta_2 = \sigma_2(\Delta)\))) } \\
	& \quad\quad \texttt{(conj} \\
	& \quad\quad\quad \texttt{(\(R_{\sigma_2}\) \(\vec{v}[\Delta_1 \mapsto \Delta_2]\))} \\
	& \quad\quad\quad \texttt{\(\enforceeqpatD(\Delta_1; \Delta_2)\)))} \\
\end{aligned}
\end{equation}

\begin{theorem}
	\label{thm:compile-poly}
	Given any \(\gamma : \Gamma\) and \(\sigma\), large-enough relation instance \(R_{\sigma_1}\), value environment \(\delta_1 : \Delta_1 =  \sigma_1(\Delta)\),
	and \(\vec v\) satisfying \(\Gamma; \Delta \vdash \texttt{(\(R\) \(\vec v\))} \wt\), it follows:
	\[ \den{\texttt{(\(R_{\sigma_1}\) \(\vec v\))}}(\gamma;\delta_1;\sigma) = \den{\compile(\texttt{(\(R_{\sigma_1}\) \(\vec v\))})}(\gamma;\delta_1;\sigma) \]
\end{theorem}
\begin{proof}
	Let \(w = \den{\texttt{(\(R_{\sigma_1}\) \(\vec v\))}}(\gamma;\delta_1;\sigma)\) and aim to show:
	\begin{equation}
		\den{\compile(\texttt{(\(R_{\sigma_1}\) \(\vec v\))})}(\gamma;\delta_1;\sigma) = w
	\end{equation}

	First consider subcases of the outer \(\texttt{fresh}\) where \(\enforceeqpatD(\Delta_1;\Delta_2)\) fails.
	We know \(0\) is multiplicative annihilator, so these cases all result in weight \(0\) and have no impact on the total weight.

	Now consider subcases of the outer \(\texttt{fresh}\) where \(\enforceeqpatD(\Delta_1;\Delta_2)\) succeeds.
	Within the \(\texttt{fresh}\), our environment is \(\delta_1, \delta_2 : \Delta_1, \Delta_2\).
	Because we assume \(\enforceeqpatD(\Delta_1;\Delta_2)\) succeeds, by lemma \ref{lem:enforce_eqpat} we know \(\delta_1 \eqpatD \delta_2\).
	From this, along \(R_{\sigma_1}\) and \(R_{\sigma_2}\) being large-enough by assumption,
	we know:
	\begin{equation}
		\den{\texttt{(\(R_{\sigma_1}\) \(\vec{v}\))}}(\gamma;\delta_1;\sigma) = \den{\texttt{(\(R_{\sigma_2}\) \(\vec{v}[\Delta_1 \mapsto \Delta_2]\))}}(\gamma;\delta_2;\sigma) = w
	\end{equation}
	where \(R_{\sigma_2}\) is the smallest large-enough instance of \(R\), as determined in the definition of \(\compile\).
	We know \(\vec{v}[\Delta_1 \mapsto \Delta_2]\) does not include any variables that occur in \(\delta_1\),
	so the \(R_{\sigma_2}\) relation call has the same weight under the larger environment \(\delta_1, \delta_2\):
	\begin{equation}
		\den{\texttt{(\(R_{\sigma_2}\) \(\vec{v}[\Delta_1 \mapsto \Delta_2]\))}}(\gamma;\delta_2;\sigma)
		= \den{\texttt{(\(R_{\sigma_2}\) \(\vec{v}[\Delta_1 \mapsto \Delta_2]\))}}(\gamma;\delta_1,\delta_2;\sigma)
		= w
	\end{equation}
	As mentioned, we assume the \(\enforceeqpatD\) call succeeds.
	Thus:
	\begin{multline}
		\den{\texttt{(conj (\(R_{\sigma_2}\) \(\vec{v}[\Delta_1 \mapsto \Delta_2]\)) \(\enforceeqpatD(\Delta_1; \Delta_2)\))}}(\gamma;\delta_1,\delta_2,\sigma) \\
		= w * 1 = w
	\end{multline}

	Our outer fresh is a sum of \(0\)s and \(w\)s. Because our semiring is associative, has idempotent \(+\), and additive identity \(0\), we have:
	\begin{multline}
		\den{\texttt{(fresh ((\(\delta_2 : \Delta_2 = \sigma_2(\Delta)\))) \(\dots\))}}(\gamma;\delta_1;\sigma) \\
		= (w + \dots) + (0 + \dots) = w + 0 = w
	\end{multline}
\end{proof}

Theorem \ref{thm:compile-poly} assures us that compiled large-enough relation calls work as expected.
Theorem \ref{thm:eqpat-weight} assures us that we can construct large-enough relations to begin with.
In practice, not all polymorphic relations need to be large-enough.
When possible, we prefer compiling relation calls as outlined here, but we can always resort to monomorphization when dealing with relation calls that are not large-enough.
As such, this method is not truly ``non-monomorphizing.''
That said, anecdotally many relation calls and definitions \emph{do} end up being large-enough; this method is non-monomorphizing ``most of the time.''

\chapter{Related Work}

\section{miniKanren}

The miniKanren\cite{byrd2012} family of languages consists of relational programming languages that are often shallowly-embedded within a host language (such as Scheme).
Relations in miniKanren are defined in terms of \emph{goals}, which usually include: value unification, fresh variable introduction, and program branching.
Languages in the miniKanren family often excel at code generation based on relational interpreters, with quine generation being a particularly compelling example.
As the name may imply, semiringKanren draws extensive inspiration from miniKanren.
In particular, semiringKanren draws most of its concrete syntax from the the microKanren\cite{hemann2023} variant.

One major difference between semiringKanren and miniKanren is that semiringKanren programs are typed, whereas miniKanren programs are untyped.
That said, typed variants of miniKanren do exist, such as OCanren\cite{kosarev2018} and typedKanren\cite{kudasov2024}.

Some variants of miniKanren support constraints.
In particular, type constraints such as \(\texttt{symbolo}\) and \(\texttt{numbero}\),
and disequality constraints (usually notated with \(\texttt{=/=}\)) are relatively common.
Values in semiringKanren are typed, so semiringKanren avoids the need for specific type constraints.
Disequality constraints are supported by semiringKanren, and are considerably easier to implement than in other miniKanren variants.

Implementations of miniKanren are usually based on top-down searches based on filtering and combining streams of potential solutions.
With semiringKanren, we use a different approach based on bottom-up fact generation, where new facts are gradually derived from existing ones.

As of writing, semiringKanren does not support recursive data types, and has not generated any quines.

\section{Datalog}

Datalog\cite{green2007} is a family of relational programming languages based on finding the \emph{least fixed point} of a database of facts.
Relations in datalog are defined in terms of Horn clauses (disjunctions of conditions that are used to determine a fact).
Datalog queries are executed by repeatedly deriving new facts using these clauses until no new facts are derived, in a process called \emph{bottom-up evaluation}.
This has the advantage of being able to compute queries that may not terminate in top-down relational languages.
The evaluation semantics of semiringKanren is directly based on this bottom-up approach.

While many models of relational algebra are boolean (relations either directly succeed or directly fail), some datalog variants introduce the concept of \emph{annotated relations}.
These generalize relations to \(K\)-relations, where facts are annotated with elements from a commutative semiring.
These annotations can be used to represent, for example, bag relations (where a fact may occur multiple times), or probabilistic databases (where facts may occur with some probability).
In particular, these annotations have been used to track \emph{provenance}, or the reasoning behind the derivation of each fact.
The implementation of semiringKanren is designed to support computation over any valid semiring.

Dyna\cite{francislandau2024} is a descendent of datalog supporting \emph{aggregators}, which permit programmers to choose from a larger set of operations on weights.

Traditionally, datalog has only a string or symbol type.
Other types are constructed by defining unary relations that only succeed for a subset of symbol values\cite{fruhwirth1991}.
In particular, this permits subtyping relations.
This differs from semiringKanren, which uses algebraic data types.
Some datalog implementations may also support numeric types, such as integers or floating-point numbers.

\section{Weighted Logic Languages}

Datalog is not the only relational programming language that supports annotating facts with weights.
In particular, there are variants of miniKanren that support probabilistic logic.
The slpKanren\cite{byrd2013} variant is based on stochastic logic, which generalizes stochastic grammars and requires disjunctions to be annotated with probabilities.
The probKanren\cite{zinkov2021} variant allows probabilistically drawing variables from both continuous and discrete distributions,
and performs probabilistic inference using a Sequential Monte Carlo method.
A variant based on weighted model counting\cite{donahue2022} also exists, which tracks the likelihood of facts across possible distributions.
Of these variants, semiringKanren is most similar to slpKanren (when computing over the usual semiring for real numbers).

Probabilistic logic programming has also been widely explored outside of the miniKanren family of languages\cite{riguzzi2022}.
Probabilistic variants of Prolog include ProbLog, Probabilistic Horn Abduction/Independent Choice Logic, and PRISM,
all of which are based on annotating disjunctions with probabilities, and have the same expressive power.
Bayesian Logic Programs, CLP(BN), and the Prolog Factor language are other probabilistic Prolog variants.

\section{Pixel Arrays}

Pixel arrays\cite{spivak2017} are a method of approximating solutions to nonlinear systems of equations.
Broadly, the equations are first discretized, then plotted in arrays.
Then, a generalized form of matrix multiplication combines these arrays to reach a final result.
More specifically, it aligns the matrices according to unexposed variables, multiplies them element-wise, then sums over the unexposed variable.
After the process is complete, the resulting array is a plot of potential solutions, up to some level of precision depending on the initial discretization.

This method is similar to how goals in semiringKanren are evaluated.
In semiringKanren, using a fresh variable to call multiple relations in conjunction is mathematically equivalent to the generalized array multiplication method used for pixel arrays.
However, the pixel array method is a discrete approximation of continuous math, and in particular may produce false positives.
Because semiringKanren computes algebraic data types that are already discrete, it has no such accuracy issues.
Furthermore, it can express more complicated relational joins than directly aligning variables.

\section{SAT Compilation}

The semiringKanren language supports evaluation via SAT solving, which can greatly improve efficiency\cite{volkov2025}.
FormuLog\cite{bembenek2018} is a datalog variant that uses SMT solvers to evaluate logical formulas, as an extension to the standard bottom-up datalog evaluation strategy.
Other relational programming engines that support SAT evaluation include Kodkod\cite{torlak2007}, Paradox, and MACE2.

\section{Polymorphic Logic Languages}

We call a typed logic programming language \emph{polymorphic} when it can define a single relation that can operate on several types.
Polymorphic languages usually fall into two categories: languages with \emph{parametric} polymorphism, or languages with \emph{ad-hoc} polymorphism.
Languages with parametric polymorphism use \emph{type variables} to represent unknown types.
In general, operations other than move and copy are not permitted on values of unknown type.
Languages with ad-hoc polymorphism usually have some means of detecting the type of values, and using this information to potentially perform different operations depending on the type.

Relational programming languages supporting parametric polymorphism include Flix\cite{madsen2020} (which embeds datalog), and functional IncA\cite{pacak2022} (which extends datalog with functions).
Both of these languages use \emph{monomorphization} to implement polymorphism,
which means generating a version of each function or relation for each type it gets called with.
We also support parametric polymorphism in semiringKanren, but do so \emph{without} monomorphization---we believe this to be first for bottom-up relational programming languages.
Note that Flix also supports computing over arbitrary semirings.

Datalog variants frequently exhibit ad-hoc polymorphism.
We have already discussed how datalog programs that implement types as unary relations often exhibit a form of subtyping.
Implementations such as PolyDatalog\cite{atzeni2010} extend this to more sophisticated object-relational type systems.

The Mercury\cite{somogyi1996} relational programming language supports parametric polymorphism by passing along specialized unification and comparison relations for each type variable
when a polymorphic relation is called.
This is different from monomorphization, as only these unification and comparison relations are specialized.
While the canonical implementation of Mercury is based on top-down search,
Mercury has also been implemented on the bottom-up Aditi\cite{vaghanl1994} deductive database system.
How exactly the implementation of polymorphism was translated is unclear,
but this may be a prior instance of non-monomorphizing polymorphism for a bottom-up relational programming language.

Note that many top-down relational programming languages, including miniKanren, are untyped and thus support a form of polymorphism by default.

\section{Other Related Work}

Tabling\cite{tamaki1986} is an approach used by top-down relational programming languages in which goal results are memoized.
This allows top-down languages to avoid getting stuck searching ineffectual paths, and overall function more similarly to bottom-up languages.
Top-down languages are often untyped/polymorphic by default, so this approach can achieve the benefits of both polymorphic and bottom-up languages.

The paper ``Testing Polymorphic Properties''\cite{bernardy2010} outlines a method to test polymorphic functions in languages with parametric polymorphism
by only testing a single monomorphic instance of the function.
While our work is different, we also show that only a single monomorphic instance is needed when using a polymorphic relation.

\chapter{Conclusions and Future Work}

We have presented semiringKanren, a polymorphic bottom-up weighted relational programming language.
This work builds on the state-of-the-art of relational programming in two ways:
semiringKanren is the first bottom-up variant of miniKanren,
and semiringKanren is the first bottom-up relational programming language to support parametric polymorphism without monomorphization.
In particular, we present a denotational semantics for semiringKanren based on operations on multidimensional arrays with elements drawn from a commutative semiring.
We show how these semantics can be extended to support polymorphism.
Finally, we introduce the notions of \emph{equality patterns} and \emph{large-enough instances} of polymorphic relations,
and show how these can be used to compile polymorphic programs into non-polymorphic ones, with minimal monomorphization.

We feel that the miniKanren approach to relational programming lends naturally to these developments.
In particular, microKanren-style \(\texttt{conj}\)- and \(\texttt{disj}\)-based syntax leads to a straightforward denotational semantics.
Furthermore, the flexibility to freely nest goals makes it straightforward to express the nontrivial patterns used when compiling polymorphic programs into non-polymorphic ones.

We hope our new method of implementing polymorphism may find applications in relational languages beyond semiringKanren.
We believe it may offer efficiency gains for polymorphic programs by reducing the need to recalculate the same relation at different types.
We also believe it can help enable ``separate evaluation'' of programs, so relations can be pre-computed and reused without reevaluation.

The OCaml implementation of semiringKanren is available on GitHub\footnote{semiringKanren repository: \url{https://github.com/sporkl/semiringkanren}}.

\section{Future Work}

\subsection{Large-Enough Relation Instances}

In this work, we present a method of counting type variable occurrences to determine when a relation is large-enough, but acknowledge that it is not sufficient in all cases.
It remains to find a fully-general method for calculating the necessary type variable sizes for a relation to be large-enough.
As we have briefly discussed, such a method may involve relation calls providing information about their own argument size requirements during the calculation process.
To handle this calculation when there are recursive calls, a fixpoint process may be necessary.

\subsection{Higher-Order Relations}

Relational programming languages that support higher-order relations (relations that take other relations as arguments) are relatively rare.
Existing work in the space includes \(\lambda\)Prolog\cite{felty1988} and \(\lambda\)Kanren\cite{weixi2020}.

With the \(\texttt{option-map}\) example, we have shown that higher-order relations can be encoded as products if they are used affinely (they can be called either zero times or once).
The power of this technique is not immediately clear.
In particular, once a relation is encoded as a product, it may be possible to effectively duplicate it (likely using a new language construct) to get more calls.

\subsection{Recursive Types}

Thus far, semiringKanren does not support recursive types.
This is very limiting.
For example, semiringKanren cannot represent arbitrary lambda calculus terms and thus perform the traditional quine generation of miniKanren languages.
In practice, recursively-typed values may be arbitrarily large, which seems incompatible with the finite array semantics used by semiringKanren.

We are aware of three potential ways to bring recursive data types to semiringKanren.
First, if semiringKanren gains full support for higher-order relations,
it may become possible to express recursive types using Church-style encodings (given that semiringKanren already supports polymorphism).

Secondly, we may be able to add support for \emph{gas-recursive types}, which are recursive types equipped with an additional size parameter.
This may make it possible to calculate maximum sizes for types, so we can find an upper bound on the sizes of the relation arrays.
This is analogous to the ``vector'' type common in dependently-typed programming languages, where the type holds a number recording the length.

Finally, it may be possible to use a system based on defunctionalization and refunctionalization to eliminate recursive data types, as in ``Exact Recursive Probabilistic Programming''\cite{chiang2023}.
From a preliminary investigation, it appears possible to translate the de/refunctionalized programs presented in that paper to semiringKanren.
It remains to be seen if the approach fully generalizes to ``derelationalization'' and ``rerelationalization'' for arbitrary semiringKanren programs.

\subsection{Semiring Fixpoint Convergence}

We have not yet investigated termination in semiringKanren.
Relatedly, we do not know if semiringKanren is Turing-complete.
These properties may depend on the specific semiring in use.

Properties of semirings have been widely explored for datalog.
For example, the base datalog language is known to always terminate and can express any polynomial-time algorithm\cite{madsen2020}.
Datalog variants computing over provenance semirings are also known to terminate\cite{green2007}.
Furthermore, certain algebraic semiring properties map to convergence properties in datalog\cite{khamis2024}.
We have not investigated whether the same properties apply to semiringKanren.

In semiringKanren, we can sometimes negate goals.
For semirings that are also rings (have additive inverses), we can express ``not'' with \(\texttt{(factor -1)}\).
For the usual semiring over real numbers, we can express disequality as follows:
\begin{equation}
\begin{aligned}
	& \texttt{(defrel (=/= (\(x : \alpha\)) (\(y : \alpha\)))} \\
	& \quad \texttt{(disj} \\
	& \quad\quad \texttt{(factor 1)} \\
	& \quad\quad \texttt{(conj} \\
	& \quad\quad\quad \texttt{(factor -1)} \\
	& \quad\quad\quad \texttt{(== \(x\) \(y\)))))} \\
\end{aligned}
\end{equation}
This roughly reads as ``succeed always, but negate the success when \(x = y\).''
Manually computing, this gives us an off-diagonal matrix as expected for \(\texttt{=/=}\).
While this particular example terminates, programs with negation in general may not.
The ``meaning'' of negation is unclear here.
Relational programming languages supporting negation are uncommon, and use a variety of approaches.
Negation has been implemented in Datalog using a three-valued semiring, and with ``stratified'' evaluation approaches (where different rules are evaluated at different times).
Several miniKanren variants support \(\texttt{noto}\) goals, including stableKanren\cite{guo2023} which is based on stable model semantics,
and minnaKanren\cite{donahue2023} which uses a system of constraints.

\subsection{SMT Solving}

We have shown that semiringKanren programs can be translated into SAT problems for efficient evaluation over the boolean semiring\cite{volkov2025}.
We believe this approach can be extended to use SMT solvers, for evaluation over other semirings.

\subsection{Applications}

It remains to see where semiringKanren may be effectively applied.
We believe semiringKanren may be a strong fit for probabilistic modelling.
It may also be worthwhile to explore existing applications of datalog, and see which may benefit from the greater expressivity of miniKanren-style syntax.

\subsection{Alternative Evaluation Strategies}

We have shown evaluation strategies for semiringKanren based on operations on semiring arrays.
In prior work, we have shown an evaluation strategy based on translation to SAT\cite{volkov2025}.
We believe there are other evaluation strategies worth exploring.

For semiringKanren's array-based evaluation, the arrays it uses are highly structured.
The diagonal and off-diagonal arrays used for equality and disequality are a clear example.
Thus, using dense arrays is likely an unnecessary use of space.
Methods exist that exploit these sorts of structures in arrays/tensors to enable parallelization and achieve significant speedups\cite{ghorbani2025}.
Implementing a backend based on such a system may make larger problems and programs tractable for semiringKanren.

We already have an practical method for compiling semiringKanren programs to SAT problems.
We can compile boolean formulas into quantum circuits, thus we can use Grover's quantum database search algorithm\cite{grover1996} to perform this evaluation.

\printbibliography[heading=bibintoc]

\end{document}